\documentclass[journal]{IEEEtran}
\usepackage{amsmath,amsfonts,amsthm,amssymb}
\usepackage{algorithmic}
\usepackage{algorithm}
\usepackage{enumitem}
\usepackage{array}
\usepackage[caption=false,font=normalsize,labelfont=sf,textfont=sf]{subfig}
\usepackage{textcomp}
\usepackage{stfloats}
\usepackage{multirow}
\usepackage{url}
\usepackage{verbatim}
\usepackage{booktabs}
\usepackage{graphicx}
\usepackage{cite}
\usepackage{xcolor}
\usepackage{cancel}
\usepackage[scr=boondoxo,scrscaled=1.05]{mathalfa}
\usepackage{bm}

\usepackage{tikz}
\usetikzlibrary{math,arrows.meta,spy,automata,intersections,decorations,arrows,decorations.markings,decorations.footprints,fadings,calc,trees,mindmap,shadows,decorations.text,patterns,positioning,shapes,patterns,decorations.pathmorphing,patterns.meta,fit,3d,plotmarks,svg.path}

\makeatletter
\tikzset{hatch distance/.store in=\hatchdistance,hatch distance=5pt,hatch thickness/.store in=\hatchthickness,hatch thickness=5pt}

\pgfdeclarepatternformonly[\hatchdistance,\hatchthickness]{north east hatch}% name
    {\pgfqpoint{-\hatchthickness}{-\hatchthickness}}% below left
    {\pgfqpoint{\hatchdistance+\hatchthickness}{\hatchdistance+\hatchthickness}}% above right
    {\pgfpoint{\hatchdistance}{\hatchdistance}}%
    {
        \pgfsetcolor{\tikz@pattern@color}
        \pgfsetlinewidth{\hatchthickness}
        \pgfpathmoveto{\pgfqpoint{-\hatchthickness}{-\hatchthickness}}       
        \pgfpathlineto{\pgfqpoint{\hatchdistance+\hatchthickness}{\hatchdistance+\hatchthickness}}
        \pgfusepath{stroke}
    }
\pgfdeclarepatternformonly[\hatchdistance,\hatchthickness]{north west hatch}% name
    {\pgfqpoint{-\hatchthickness}{-\hatchthickness}}% below left
    {\pgfqpoint{\hatchdistance+\hatchthickness}{\hatchdistance+\hatchthickness}}% above right
    {\pgfpoint{\hatchdistance}{\hatchdistance}}%
    {
        \pgfsetcolor{\tikz@pattern@color}
        \pgfsetlinewidth{\hatchthickness}
        \pgfpathmoveto{\pgfqpoint{\hatchdistance+\hatchthickness}{-\hatchthickness}}
        \pgfpathlineto{\pgfqpoint{-\hatchthickness}{\hatchdistance+\hatchthickness}}
        \pgfusepath{stroke}
    }
\makeatother

\pgfdeclarepattern{
  name=hatch,
  parameters={\hatchsize,\hatchangle,\hatchlinewidth},
  bottom left={\pgfpoint{-.1pt}{-.1pt}},
  top right={\pgfpoint{\hatchsize+.1pt}{\hatchsize+.1pt}},
  tile size={\pgfpoint{\hatchsize}{\hatchsize}},
  tile transformation={\pgftransformrotate{\hatchangle}},
  code={
    \pgfsetlinewidth{\hatchlinewidth}
    \pgfpathmoveto{\pgfpoint{-.1pt}{-.1pt}}
    \pgfpathlineto{\pgfpoint{\hatchsize+.1pt}{\hatchsize+.1pt}}
    \pgfpathmoveto{\pgfpoint{-.1pt}{\hatchsize+.1pt}}
    \pgfpathlineto{\pgfpoint{\hatchsize+.1pt}{-.1pt}}
    \pgfusepath{stroke}
  }
}

\pgfdeclarepattern{
  name=pLines,
  parameters={\hatchsize,\hatchangle,\hatchlinewidth},
  bottom left={\pgfpoint{-.1pt}{-.1pt}},
  top right={\pgfpoint{\hatchsize+.1pt}{\hatchsize+.1pt}},
  tile size={\pgfpoint{\hatchsize}{\hatchsize}},
  tile transformation={\pgftransformrotate{\hatchangle}},
  code={
    \pgfsetlinewidth{\hatchlinewidth}
    \pgfpathmoveto{\pgfpoint{-.1pt}{-.1pt}}
    \pgfpathlineto{\pgfpoint{\hatchsize+.1pt}{\hatchsize+.1pt}}
    \pgfusepath{stroke}
  }
}

\tikzset{
  hatch size/.store in=\hatchsize,
  hatch angle/.store in=\hatchangle,
  hatch line width/.store in=\hatchlinewidth,
  hatch size=5pt,
  hatch angle=0pt,
  hatch line width=.5pt,
}

\tikzset{
  pLines size/.store in=\hatchsize,
  pLines angle/.store in=\hatchangle,
  pLines line width/.store in=\hatchlinewidth,
  pLines size=5pt,
  pLines angle=0pt,
  pLines line width=.5pt,
}

\tikzstyle{state}=[draw=black, circle, fill=cyan!30, text=black, minimum width=2em, line width=1pt]
\tikzstyle{dstate}=[draw,cyan, circle, fill=cyan!40, text=black, minimum width=2em, line width=1pt]
\tikzstyle{fixed}=[draw=black,lightgray, circle, fill=Cornsilk!40, text=black, minimum width=2em, line width=1pt]

\tikzstyle{fstate}=[draw=black, fill=red!20, text=black, minimum width=1.25em, minimum height=1.25em,preaction={fill, red!20}, pattern=hatch, pattern color=black, hatch size=7.5pt]
\tikzstyle{flstate}=[draw,black, text=black, minimum width=1.25em, minimum height=1.25em, preaction={fill, violet!30}, pattern=hatch, pattern color=black, hatch size=7.5pt, hatch angle = -45]

\tikzstyle{fdstate}=[draw,black, text=black, minimum width=1.25em, minimum height=1.25em,preaction={fill, blue!10},pattern=pLines, pattern color=blue, pLines size = 8pt, pLines angle=0]
\tikzstyle{ffixed}=[draw=black, text=black, minimum width=1.25em, minimum height=1.25em,preaction={fill,lightgray!20}, pattern=north west hatch,hatch distance=7pt,hatch thickness=.5pt,pattern color=black]
\tikzstyle{fhfixed}=[draw=black, text=black, minimum width=1.25em, minimum height=1.25em, preaction={fill, green!10},pattern = dots]
\tikzstyle{fzero}=[draw=black, text=black, minimum width=1.25em, minimum height=1.25em, preaction={fill, orange!30},pattern = bricks]
\tikzstyle{fbar}=[draw=black, fill=yellow!20, text=black, minimum width=1.25em, minimum height=1.25em]

\usepackage[svgnames]{xcolor}

\DeclareMathOperator*{\argmin}{arg\,min}
\newtheorem{prop}{Proposition}
\newtheorem{remark}{Remark}

\newcommand{\eor}{\ensuremath{\hfill\blacklozenge}}

\providecommand{\norm}[1]{\ensuremath{\left\lVert#1\right\rVert}}
\providecommand{\abs}[1]{|#1|}

\newcommand{\bbm}{\begin{bmatrix}}
\newcommand{\ebm}{\end{bmatrix}}

\definecolor{lavender}{RGB}{150, 123, 182}
\begin{document}

\title{Factor Graph Optimization for Leak Localization in Water Distribution Networks}

% of /department name/, /university/, /city/,% ADD CSIC TO THE AFFILIATION
% /postal code/ /country/ e-mail: /email/.}% <-this % stops a space

\author{Paul~Irofti,~\IEEEmembership{Member,~IEEE,}, Luis~Romero-Ben,
Florin~Stoican,~\IEEEmembership{Member,~IEEE,}
and~Vicenç~Puig% <-this % stops a space
\thanks{Paul Irofti is with LOS-CS-FMI,
University of Bucharest, 
e-mail: {paul@irofti.net}}%
\thanks{Luis Romero-Ben and Vicenç Puig are with the Institut de Robòtica i Informàtica Industrial, CSIC-UPC, Llorens i Artigas 4-6, 08028, Barcelona,
e-mail: \{luis.romero.ben,vicenc.puig\}@upc.edu}% <-this % stops a space
\thanks{Florin Stoican is with the  Dept. of Automation Control and Systems Engineering, Politehnica University of Bucharest,
e-mail: {florin.stoican@upb.ro}.}% <-this % stops a space
\thanks{Vicenç Puig is also with the Supervision, Safety and Automatic Control Research Center (CS2AC) of the Universitat Politècnica de Catalunya
}% <-this % stops a space
\thanks{
This work was supported in part by
the project “Romanian Hub for Artificial Intelligence - HRIA”, Smart Growth, Digitization and Financial Instruments Program, 2021-2027, MySMIS no. 334906,
and in part by the project “Sustainable and learning-based management of large-scale multi-resource systems” - SEAMLESS, ref. PID2023-148840OB-I00, and the project “Sensor fusion for real-time monitoring of leaks in water distribution networks” - FUSAGUA, no. 202550E039, both funded by MCIU/ AEI /10.13039/501100011033.
%Manuscript received April 19, 2005; revised August 26, 2015.
}
}%



\maketitle

% https://www.ieeesmc.org/publications/transactions-on-smc-systems/information-for-authors/

% Provide an informative abstract of 100 to 250 words at the head of the manuscript.
% without abbreviations, footnotes, or references; written as one paragraph, and should not contain displayed mathematical equations or tabular material;
\begin{abstract}
Detecting and localizing leaks in water distribution network systems is an important topic
with direct environmental, economic, and social impact.
Our paper is the first to explore the use of factor graph optimization techniques for leak localization in water distribution networks,
enabling us to perform sensor fusion between pressure and demand sensor readings
and to estimate the network's temporal and structural state evolution across all network nodes.
The methodology introduces specific water network factors
and proposes a new architecture composed of two factor graphs:
a leak-free state estimation factor graph
and a leak localization factor graph.
When a new sensor reading is obtained,
unlike Kalman and other interpolation-based methods,
which estimate only the current network state,
factor graphs update both current and past states.
Results on Modena, L-TOWN and synthetic networks
show that factor graphs are much faster than nonlinear Kalman-based alternatives such as the UKF,
while also providing improvements in localization
compared to state-of-the-art estimation-localization approaches.
Implementation and benchmarks are available at \url{https://github.com/pirofti/FGLL}.
\end{abstract}


%------------------------------------------------------------------
% 12 pages or shorter
% arXiv note: “This work has been submitted to the IEEE for possible publication. Copyright may be transferred without notice, after which this version may no longer be accessible.”

\section{Introduction}\label{sec:intro}
\IEEEPARstart{W}{ater} is %a crucial resource for the development and surveillance of 
essential for modern societies, with the global demand projected to grow by 55\% between 2000 and 2050 \cite{Leflaive2012}. However, inefficacies during storage, transport, treatment, and distribution reduce service quality and increase economic, environmental, and even sanitary costs \cite{Cohen2004}. A major challenge for water utilities is the appearance of leaks in water distribution networks (WDNs), which are calculated to sum up to 126 billion m\textsuperscript{3} of water worldwide per year \cite{Liemberger2019}. This justifies the interest of water utilities in effective leak management methods. Traditionally, water utilities used night flow analysis and human teams with acoustic sensors to detect and locate leaks. However, this localization approach has scalability problems. Leaks can be fixed by means of software-based methods, which monitor the entire network and reduce the search area. Such methods can be classified by their reliance on a hydraulic model of the network.

Model-based methods require well-calibrated hydraulic models to simulate the network behavior in leak scenarios, comparing the obtained data with real measurements to identify the leaks. In this category, we can find important methodologies based on sensitivity analysis \cite{Perez2014}, inverse problem \cite{Sophocleous2019} and Bayesian theory \cite{Costanzo2014}. These methods normally excel when the model is well calibrated, but this is difficult due to the complexity of real-world networks (and the hydraulic model) and potential modeling errors \cite{Blesa2018}. Advancements in data analysis and machine learning reduced the dependence on the hydraulic model, which is used mainly to generate training data. These mixed model-based/data-driven methods include algorithms such as artificial neural networks (ANNs) \cite{Sun2019}, deep learning \cite{Zhou2019} and dictionary learning \cite{Irofti2020}. Recently, purely data-driven methods have emerged as a solution for water utilities lacking a well-calibrated hydraulic model. %Unlike model-based and hybrid approaches, data-driven methods do not require a hydraulic simulator. Instead, they 
These methods locate leaks by comparing network states under leak and leak-free conditions. In particular, most of these methods are based on state interpolation, which estimates the complete network state (defined by a chosen hydraulic variable) from available measurements and the underlying graph structure~\cite{Soldevila2020, RomeroBen2022}.

Interpolation-based methods typically use nodal pressures or hydraulic heads (pressure and elevation) as state representatives, given the reduced cost and ease of installing pressure sensors compared to other sensor types and their known sensitivity to leaks. However, modern WDNs are increasingly incorporating additional sensor types due to technological advances and the decrease in ancillary costs. Therefore, our most recent research focuses on the integration of sensor fusion techniques within state estimation and leak localization approaches, enabling the exploitation of all available sources of information across the network. For example, \cite{RomeroBen2024b} presented a methodology that fuses demand and pressure information through an Unscented Kalman Filter (UKF) approach.

This article presents the first attempt to study the utilization of factor graph optimization (FGO) methods for WDN sensor fusion, state estimation, and leak localization, allowing to integrate pressure and demand measurements during the process.

\noindent\textbf{Contributions}: %hereinafter, we present various contributions:

\begin{itemize}
    \item A novel architecture is proposed to execute the state estimation and localization operations.
    \item Custom factors are designed not only to perform sensor fusion or state interpolation operations but also to exploit hidden constraints and dependencies obtained from expert knowledge.
    \item The proposed method is capable of exploiting temporal information, performing a single optimization to obtain the state estimation of a set of time instants, as well as a single leak localization result that considers all of them.
    %\item The proposed method shows that it can match or even surpass the level of performance of several state-of-the-art methods. \textcolor{red}{ADD PERFORMANCE INDICATORS WHEN THE RESULTS ARE FINISHED}
\end{itemize}

\noindent\textbf{Notation}: We represent the set of all natural numbers up to $N$ with $[N]$.
Scalars are denoted through lowercase letters, while vectors and matrices are represented in bold, employing lowercase and uppercase letters, respectively. In an arbitrary matrix $\bm{A}$, the $i$-th column is represented as $\bm{a}_i$, and $a_{ij}$ indicates the $i$-$j$ element. The notation $\bm{x}^{[t]}$ is represents the values of the vector $\bm{x}$ at time instant $t$, which is equivalent to $\bm{x}(t)$. %Nevertheless, sub-indices $h$ and $q$ indicate that the corresponding variable belongs either to the head or flow estimation process respectively. 
Within estimation processes, the use of a hat indicates an approximation, such as $\bm{\hat{A}}$ being the approximation of $\bm{A}$. Other important notations include $\left(\bm{x}\right)^{y}$, which means that each component of vector $\bm{x}$ is raised to the power of $y$, and $|\cdot|$, which denotes cardinality if applied to sets, and the absolute value otherwise. The notation $\bm{x}_n$ is used to indicate a column vector with $n$ elements, all of them having a value of $x\in\mathbb{N}_0$. 

\noindent\textbf{Outline}:
In Section~\ref{sec:prelim} we introduce preliminaries for leak localization and factor graph optimization on which we base the methodology from Section~\ref{sec:methodology}.
This includes the specific water network factors and the main algorithm.
Section~\ref{sec:studies} describes the water distribution networks studied
on which we perform the numerical simulations in Section~\ref{sec:sim}.
We draw conclusions about the proposed methodology in Section~\ref{sec:conclusion}.

%------------------------------------------------------------------
\section{Preliminaries}\label{sec:prelim}

In this section, we briefly introduce the necessary concepts for leak localization and factor graph optimization.

\subsection{Leak management}\label{sec:prelim_leak}

The delivery of water to the users is carried out through the water distribution network, i.e. a system of pressurized pipes that connect the supply tanks with the consumption points. During the supply process, water can be lost due to the appearance of leaks, which can be caused by a deficient state of the infrastructure, material defects, installation or maintenance errors, harsh environmental conditions, excessive water pressure, etc. Therefore, water utilities use leak management methods to handle the appearance of leaks to minimize water losses and service degradation. The leak management problem can be divided into three subproblems: optimal sensor placement, leak detection, and leak localization. In this work, we will focus on the problem of leak localization, so we assume that a set of optimally placed pressure, flow and/or demand sensors have been installed in the WDN, and that an ideal leak detection algorithm is already operational. 

\subsubsection{Background}

To properly present the leak management problem, a water distribution network is modeled as a graph $\mathcal{G} = (\mathcal{V},\mathcal{E})$, where $\mathcal{V}$ represents the set of nodes (junctions/reservoirs of the WDN) and $\mathcal{E}$ is the set of edges (pipes). The $i$-$th$ node of the graph is indicated by $\mathscr{v}_i\in\mathcal{V}$, whereas the edge connecting $\mathscr{v}_i$ (source) to $\mathscr{v}_j$ (sink) is given as $\mathscr{e}_k = \mathscr{e}_{ij} \in \mathcal{E}$. The graph is then composed of $n = |\mathcal{V}|$ nodes and $m = |\mathcal{E}|$ edges. Concerning the network's physical variables, we define $\bm{p}\in\mathbb{R}^n$ as the nodal pressure vector, so $\bm{h} = \bm{p}+\bm{e}$ is the nodal hydraulic head, where $\bm{e}\in\mathbb{R}^n$ is the nodal elevation vector. In addition, $\bm{q}\in\mathbb{R}^m$ is the pipe flow vector, and $\bm{d}\in\mathbb{R}^n$ is the vector of nodal demands. The relationship between nodal heads and pipe flows is given through an experimentally-obtained equation of the form:

\begin{equation}\label{eq:flow-head}
    |h_i - h_j| = \tau_{k}q_{k}^{\nu},
\end{equation}

\noindent where $h_i$ is the hydraulic head at $\mathscr{v}_i$, $q_k=q_{ij}$ is the flow traversing $\mathscr{e}_k$ from $\mathscr{v}_i$ to $\mathscr{v}_j$, and $\tau_{k}=\tau_{ij}$ and $\nu$ are the pipe resistance coefficient and the flow exponent. Our work uses the Hazen-Williams equation \cite{Brater1996}
with $\tau_{k} = 10.674 \frac{\rho_{k}}{\xi_{k}^{\nu} \psi_{k}^{4.87}}$, 
where $\rho_{k}, \xi_{k}$ and $\psi_{k}$ are the length, roughness and diameter of the pipe represented by $\mathscr{e}_k$, and $\nu=1.852$.

% With this notation, the concept of leak can be formally defined. Although leaks normally occur in network pipes, the assumption of leaks occurring in nodes is widely accepted, as virtual nodes can be created in pipes to keep this assumption. In this case, we can define the flow rate of a leak occurring at $\mathscr{v}_i$ as:

% \begin{equation}\label{eq:leak-definition}
%     q^{leak}_i=\delta_d p_i^{\chi}
% \end{equation}

% \noindent with $\delta_d$ and $\chi$ being the discharge coefficient and the pressure exponent respectively. The first depends on the physical and shape parameters of the pipe, whereas the second depends on how the leak flow responds to pressure changes.

\subsubsection{Leak localization}

Consider an arbitrary WDN whose steady-state hydraulic state is defined as $\bm{x}(t)=[\bm{h}^\top(t)\:\:\bm{q}^\top(t)\:\:\bm{d}^\top(t)]^{\top}\in\mathbb{R}^{n_x}$, where $n_x=2n+m$ and $t$ is the time instant. The installed sensors provide direct measurements of the states subset $\bm x_S(t) = \bm S \bm x(t)$, where $\bm S \in \mathbb{R}^{n_S\times n_x}$ is the binary measurement matrix encoding the measured states, $n_S$ being the total number of sensors.

The leak localization problem consists of pinpointing the location of the existing leak/s over the network. In a data-driven context, this process would only require:
data gathered from the sensors, i.e., $\bm X_S=\left[\bm{x}_S(t)\:\:\bm{x}_S(t+1)\;\ldots\;\bm{x}_S(t+T)\right]\in \mathbb{R}^{n_S\times T+1}$ with $T+1$ being the size of the data collection time window\footnote{The localization process starts after the leak is detected, and thus the data collected during the period from $t$ to $t+T$ always includes at least one leak.};
data from a leak-free reference scenario $\overline{\bm X}_{S}=\left[\bm{x}_S(t_{nom})\:\:\bm x_S(t_{nom}+1)\;\ldots\;\bm{x}_S(t_{nom}+T)\right]$, where $t_{nom}$ denotes a time instant with no leak detected, and with similar boundary conditions to $t$, i.e. demand patterns, inflow to the network, tank/s volumes, etc.;
structural information, encoded in $\mathcal{G}$.
Thus, the leak localization process can be defined by the following expression:

\begin{equation}\label{eq:localization}
    \mathcal{V}_L = \mbox{f}\left(\bm{X}_S,\overline{\bm X}_{S},\mathcal{G}\right),
\end{equation}

\noindent where $\mathcal{V}_L\subseteq\mathcal{V}$ is the set of leak candidate nodes and $\mbox{f}(\cdot)$ is the leak localization function.

% The leak localization problem is typically considered in the literature in the context of single-leak scenarios with ideal nominal reference (although uncertainty in demand patterns or physical parameters can be considered). Nevertheless, real-world events include other phenomena that hinder the localization process, such as:

% \begin{itemize}
%     \item Multi-leak scenarios, which are defined by the existence of at least two simultaneous leaks affecting the network's behaviour. 
%     \item The operation of active elements such as pumps or valves, which may disrupt the behaviour of the network and generate differences between nominal and leak scenarios.
% \end{itemize}

%\todo{Luis, Vicenç: the problem, the multi-leak case, battlenet state change example for switchable constraints}

\subsection{Factor Graph Optimization}\label{sec:prelim_fgo}

Factor graphs\cite{Kschischang01_FactorGraphSumProdAlgo} are bipartite graphs $F = (U,V,E)$ where edges $E$ connect factor nodes $U$ to variable nodes $V$.
The variable nodes are data providers to the factor nodes.
Factor nodes represent the data models
and usually require solving computationally intensive tasks
such as nonlinear optimization problems in order to produce the next data item.

What sets aside factor graph optimization algorithms
from similar state estimation approaches,
such as Kalman and Particle filters \cite{elfring2021particle}, where the entire state history is assumed to be adequately summarized by the state's value at the previous time instant,
is the fact that the solution updates all states (past and present)
unlike the others which only update the present state.
FGO problems are usually solved with the forward/backward algorithm,
visiting each graph node exactly twice (or once per pass),
updating the variables, and providing consistency across states.
A single forward/backward iteration is sufficient most of the time.
The main computation is performed by the factor nodes,
where we usually have to solve a nonlinear problem based on the neighboring variable nodes.

In the leak localization scenario,
we are interested in estimating the network state $x$
from available sensor readings
such as
demands $d$ and heads $h$.
In Figure~\ref{fig:fgo-example},
variables nodes are denoted with circles
and the factor nodes are denoted with squares.
We model the current demand $d_0$ and head $h_0$ sensor readings
as variable nodes.
The previous $x_0$ and current (estimated) water network state $x_1$ are also variables.
The red factor function (\begin{tikzpicture}[baseline=-2pt]\node[draw,minimum width=1pt,line width=1pt,red,fill=red!20, scale=1] at(0,1pt){};\end{tikzpicture}) in the left panel takes the prior state together with the new sensor readings to produce the estimated state $x_1$.
The factor can be implemented by most state estimation techniques,
such as
interpolation methods~\cite{RomeroBen2022, Irofti2023, Zhou2023},
Kalman filters~\cite{Govaers2019,Julier1997} and their combination \cite{RomeroBen2024b}.
In the right panel, each variable node is connected to a specific factor node.
In our example
the blue node (\begin{tikzpicture}[baseline=-2pt]\node[draw,minimum width=1pt,line width=1pt,blue,fill=blue!20, scale=1] at(0,1pt){};\end{tikzpicture}) uses demand measurements to update the state,
the green node (\begin{tikzpicture}[baseline=-2pt]\node[draw,minimum width=1pt,line width=1pt,SeaGreen,fill=SeaGreen!20, scale=1] at(0,1pt){};\end{tikzpicture}) estimates the heads in all nodes,
and the red in-between states node
ensures consistency between the previous and current state
based on measured versus interpolated data.

\begin{figure}
    \centering
    \includegraphics[width=\linewidth]{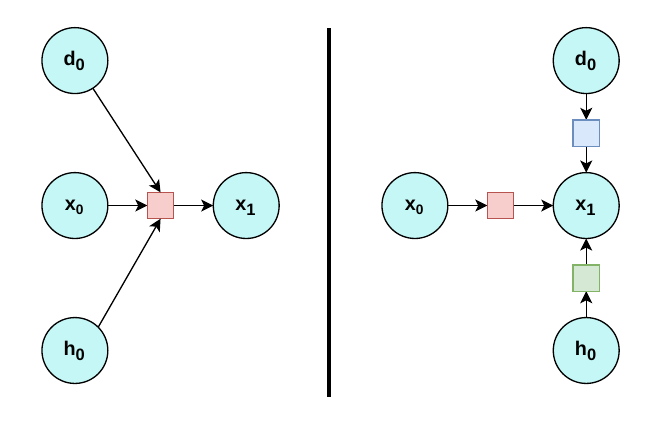}
    \caption{Factor graph example.
    Left: three variables
    (demands $d_0$, heads $h_0$ and initial state $x_0$)
    are used by the factor function depicted by the pink square to produce the next state estimate $x_1$.
    Right: the same three variables are processed now by their own factor functions whose combined output produces the next state estimate $x_1$.}
    \label{fig:fgo-example}
\end{figure}

There are multiple important aspects when designing factor graphs.
Of special interest in this paper are
marginalization~\cite{Wang23_marginalization} and switchable constraints~\cite{Li18_switchable}.

To explain the concept of marginalization, let us consider the WDN scenario, where new sensor readings lead to a new system state. In a fixed time horizon, such as an hour or a day, this would lead to $N$ states, which enable the building of the factor graph by connecting $N$ elements such as those from in~Figure~\ref{fig:fgo-example}. After performing factor graph optimization, we obtain state estimates for all $N$ states. When new sensor readings arrive, the factor graph will expand accordingly, and FGO will have to be performed again on this new graph. Even though FGO updates all system states when dealing with a broad time horizon,  older states tend to remain constant while only a small recent state window fluctuates. We can take advantage of this by resuming the older states into a single prior state, thus shrinking the factor graph and gaining faster computations with limited estimation loss. Therefore, marginalization consists of the process of marginalizing the older states into a single prior state, allowing us to move across the time domain with a fixed window size.

Switchable constraints connect variables related to the same sensor across the horizon by means of factors that perform outlier detection and zero-out inconsistent readings, blocking an affectation to the corresponding state estimate. This is especially useful to eliminate readings that are malfunctioning sensor from the measurement pool \cite{seron2011fault}, and to mitigate the effect of active elements, such as valves or pumps, which may cause sudden large jumps in the value of hydraulic variables such as pressure or flow.

%Sensors are sensible instruments that due to their environment and manufacturing
% often lead to erroneous readings either temporary or permanently
% which is an active topic in the predictive maintenance field. When hardware redundancy exists, often temporary (until a recovery test is verified) elimination from the sensor pool, is preferable \cite{seron2011fault}.
% FGO mitigates this phenomenon through switchable constraints~\cite{Li18_switchable}
% where variables corresponding to the same sensor are connected across the timeline
% through factors that perform outlier detection and zero-out inconsistent readings
% such that the corresponding state estimate is not affected. \textit{In WDN, the operation of active elements such as valves or pumps may cause sudden changes in hydraulic variables like pressure or flow. In such cases, even if sensor measurements are accurate, a large jump may still occur between values at two consecutive time steps. The possibility of including switchable constraints could be useful to avoid the discrepancies introduced by the affected measurements.}
% It is as if the sensor never provided a measurement at that time.
%\todo{Luis, give an example here of a numerically obvious jump in head and demand readings between two time instances.}

% \todo{Paul: add about $\Sigma$'s and NLS problems to help factor description below; e.g. $l_1 = l_2$ is actually an approximation within $\Sigma$ bounds}
Solving the factor graph optimization problem
is often approached through Maximum a Posteiori Inference (MAP)~\cite{Dellaert17_book}
leading to the non-linear least-squares problem
$x^\star = \arg\min_x \sum_i \norm{f_i(x_i) - z_i}_{\Sigma_i}^2$
% \begin{align}
% x^\star 
%   % &= \arg\max_x \prod_i \phi_i (x_i)  \\
%   % &= \arg\max_x \exp\left\{-\frac12 \norm{f_i(x_i) - z_i}_{\Sigma_i}^2 \right\} \\
%   &= \arg\min_x \sum_i \norm{f_i(x_i) - z_i}_{\Sigma_i}^2
%   \label{nls_problem}
% \end{align}
where $x$ are the unknown state variables,
$z$ are the (sensor) measurements
with associated noise covariances $\Sigma$,
and $f$ the factors.
In practice, we search for a proper $\Sigma$
when the noise model is unknown.

%------------------------------------------------------------------
\section{Methodology}\label{sec:methodology}

In this section, we present our methodology regarding interpolation, the proposed WDN-specific factors together with the resulting factor graphs,
and our leak localization algorithm.

\subsection{Interpolation}\label{sec:interpolation}

In the context of data-driven leak localization, interpolation-based techniques are composed of two main stages: state estimation and leak localization. The former retrieves the complete hydraulic state of the network, defined by a certain hydraulic variable (or set of variables), from the available measurements and structural information. The latter compares leak and leak-free data in order to pinpoint the leak location. Recalling \eqref{eq:localization}, it can be adapted as 
$\mathcal{V}_L = \mbox{f}\left(\bm{X}_S,\overline{\bm X}_{S},\mathcal{G}\right) = \mbox{f}_L\left(\mbox{f}_I(\bm{X}_S,\mathcal{G}),\mbox{f}_I(\overline{\bm X}_{S},\mathcal{G})\right)$
for the case of interpolation-based strategies,
% \begin{equation}\label{eq:interp_loc}
%     \mathcal{V}_L = \mbox{f}\left(\bm{X}_S,\overline{\bm X}_{S},\mathcal{G}\right) = \mbox{f}_L\left(\mbox{f}_I(\bm{X}_S,\mathcal{G}),\mbox{f}_I(\overline{\bm X}_{S},\mathcal{G})\right),
% \end{equation}
where $\mbox{f}_I(\cdot)$ is the interpolation function and $\mbox{f}_L(\cdot)$ is the localization function. Note that $\hat{\bm X}=\mbox{f}_I(\bm{X}_S,\mathcal{G})$, with $\hat{\bm X}$ being the reconstructed version of $\bm X$ (and an analogue decomposition can be applied in the leak-free case).

\begin{remark}
    In practice, most interpolation-based methods use nodal hydraulic head (or pressure) as the WDN state representative variable. As explained in Section \ref{sec:intro}, this is justified by the advantages of pressure sensors over other types of metering devices, such as lower cost and easier installation, as well as the high sensitivity of pressure to leaks. Thus, from now on, we consider that $\bm x (t) = \bm h(t)$. \eor
\end{remark}

Previously, we developed an interpolation technique called Graph-based State Interpolation (GSI) \cite{RomeroBen2022}. This method, following the structure of $\mbox{f}_I(\bm{X}_S,\mathcal{G})$, only requires hydraulic head data from the sensorized nodes and the network structure and consists of solving the quadratic programming problem:

%--------------------
\begin{align}\label{eq:GSI_opt}
\bm{\hat{h}}=\argmin_{\bm{h}, \gamma} \quad & \frac{1}{2}\big[\bm{h}^T\bm{L}\bm{D}^{-2}\bm{L}\bm{h}+\chi \gamma ^2\big],\\
\label{eq:GSI_opt_b}\textrm{s.t.} \quad & \hat{\bm{B}} \bm{h}\leq \gamma \cdot\bm{1}_{m},\ \gamma > 0,\ \bm{S}_p\bm{h}=\bm{h}_{s},  
\end{align}
%--------------------

\noindent where $\bm{h}=\bm{h}(t)\in \mathbb{R}^{n}$ is a notation relaxation to improve readability. %, as GSI operates over measurements of individual time instants, therefore providing a single estimation for the considered instant. 
Moreover, $\hat{\bm{h}}$ is the reconstructed state (hydraulic heads) and $\bm{L}=\bm D-\bm W$ is the Laplacian matrix of $\mathcal{G}$, which is a weighted graph with an associated weighted adjacency matrix $\bm W\in \mathbb{R}^{n\times n}$ (defining $w_{ij}=\frac{1}{\rho_k}$ if $\mathscr{e}_k=\mathscr{e}_{ij}\in\mathcal{E}$, and $w_{ij}=0$ otherwise) and a degree matrix $\bm D\in \mathbb{R}^{n\times n}$. $\hat{\bm{B}}\in \mathbb{R}^{m\times n}$ is an approximate incidence matrix of $\mathcal{G}$, calculated as explained in Algorithm 1 of \cite{Irofti2023} and $\gamma\in \mathbb{R}$ is a slack variable for the inequality constraints. $\bm S_p\in \mathbb{R}^{n_s\times n}$ is the sensorization matrix, which encodes the sensor locations (with $n_s$ being the number of pressure sensors), $\bm h_s\in \mathbb{R}^{n_s}$ is the vector of head measurements and $\chi\in \mathbb{R}$ is a weight measuring the relative importance between the two sub-objectives. In summary, GSI pursues two goals: (1) minimization of the square difference between the state of a node and its predicted state by a weighted linear combination of its neighbors' state, and (2) maximization of the similarity between the flow directionality encoded in the reconstructed state and the directionality indicated by $\hat{\bm{B}}$, by means of minimization of the squared slack variable $\gamma^2$. 

The FGO-based methodology proposed in this article requires an interpolation operation that can be integrated into the FGO structure. Therefore, we present a novel version of GSI, denoted as Fast GSI or FGSI.

\begin{prop}
    The solution of the GSI optimization problem in \eqref{eq:GSI_opt} can be approximated by a single matrix multiplication:
    \begin{equation} \label{eq:proposition1_1}
        \hat{\bm{h}} \approx (\mu_L\bm{L}\bm{D}^{-2}\bm{L} + \bm{S}_p^{\top}\bm{S}_p)^{-1}\bm{S}_p^{\top}\bm{h}_s = \bm{P}_s\bm{h}_s,
    \end{equation}
    \noindent where $\mu_L\in \mathbb{R}$ is a relative weight and $\bm{P}_s \in \mathbb{R}^{n\times n_s}$ is the computed interpolation matrix.
\end{prop}
\begin{proof}
    In order to obtain an approximate solution to the GSI problem, the constraints in \eqref{eq:GSI_opt} can be relaxed or removed to lead to an unconstrained optimization problem. First, the directionality terms are eliminated from the problem, removing the inequality constraints and the objective terms related to $\gamma$. Then, the sensorization constraint is relaxed, converting the equality constraint into an objective term minimizing the differences between measurements and measured states, in a similar fashion to penalty methods \cite{Eason2015}:
    \begin{multline}\label{eq:proposition4_proof_1}
            % \begin{split}
            (\bm{S}_p\bm{h} - \bm{h}_s)^{\top}(\bm{S}_p\bm{h} - \bm{h}_s) = (\bm{h}^{\top}\bm{S}_p^{\top} - \bm{h}_s^{\top})(\bm{S}_p\bm{h} - \bm{h}_s)=\\
            \bm{h}^{\top}\bm{S}_p^{\top}\bm{S}_p\bm{h} - 2\bm{h}_s^{\top}\bm{S}_p\bm{h} + \bm{h}_s^{\top}\bm{h}_s.
            % \end{split}
        \end{multline}
        Thus, the solution of the constrained problem in \eqref{eq:GSI_opt} is approximated by the solution of the next unconstrained problem
        $\hat{\bm{h}} \approx \argmin_{\bm{h}} \: \frac{1}{2}\Big[\bm{h}^{\top}\big(\mu_L\bm{L}\bm{D}^{-2}\bm{L} + \bm{S}_p^{\top}\bm{S}_p\big)\bm{h}\Big] - \bm{h}_s^{\top}\bm{S}_p\bm{h}$,
    where the term $\frac{1}{2}\bm{h}_s^{\top}\bm{h}_s$ from \eqref{eq:proposition4_proof_1} is not included because it is constant for a certain $\bm{h}_s$. Moreover, $\mu_L$ is a parameter that regulates the relative influence of the network structure, encoded in $\bm{L}\bm{D}^{-2}\bm{L}$, and the sensorization structure, represented by $\bm{S}$. The optimal solution to this problem is computed by setting the derivative with respect $\bm{h}$ to zero
    $\left(\mu_L\bm{L}\bm{D}^{-2}\bm{L} + \bm{S}_p^{\top}\bm{S}_p\right)\bm{h}_F - \bm{S}_p^{\top}\bm{h}_s = \bm{0}_n$,
    leading to the result
    $\bm{h}_F = \left(\mu_L\bm{L}\bm{D}^{-2}\bm{L} + \bm{S}_p^{\top}\bm{S}_p\right)^{-1}\bm{S}_p^{\top}\bm{h}_s$
    % \begin{equation}\label{eq:proposition4_proof_3}
    % \begin{split}
    %     \left(\mu_L\bm{L}\bm{D}^{-2}\bm{L} + \bm{S}_p^{\top}\bm{S}_p\right)\bm{h}_F - \bm{S}_p^{\top}\bm{h}_s = \bm{0}_n, \\
    %     \bm{h}_F = \left(\mu_L\bm{L}\bm{D}^{-2}\bm{L} + \bm{S}_p^{\top}\bm{S}_p\right)^{-1}\bm{S}_p^{\top}\bm{h}_s,
    %     \end{split}
    % \end{equation}
    where $\bm{h}_F \approx \hat{\bm{h}}$ is the optimal head estimate from \eqref{eq:proposition1_1} that approximates  \eqref{eq:GSI_opt}. 
\end{proof}

\begin{remark}\label{remark:FGSI}
    FGSI leads to an additional advantage with respect to GSI: the explicit interpolation of pressure residuals. GSI includes directionality terms, which cannot be directly related to pressure residuals, as some flow directions may change between leak-free and leak scenarios. FGSI avoids these terms, and its multiplicative nature facilitates its use for residual interpolation. \eor
\end{remark}
% \todo{Luis, Vicenç: describe GSI and other interpolation equations like H-W?}

\subsection{Water Distribution Network Factors}\label{sec:water_factors}

The proposed methodology is defined by the factor graph presented in Figure \ref{fig:fgo-arch}. Within this architecture, a handful of factors are defined to encode the relationships between different states and fixed values. This set of states contains the head state $\bm h^{[t]}$ (with $t$ indicating the time instant)\footnote{Note that, for clarity, we substitute the $\bm x(t)$ notation for $\bm x^{[t]}$.}, which represents the estimation of the nodal hydraulic head vector; the demand state $\bm d^{[t]}$, which represents the estimation of the nodal consumption vector; and the localization state $\bm l^{[t]}$, which represents the estimation of pressure residuals, used in this case as a localization likelihood index. We discuss the designed factors in the following.

\subsubsection{Temporal evolution}\label{subsubsec:temp_factor}

This factor is indicated through a red square (\begin{tikzpicture}[baseline=-2pt]\node[fstate, scale=.8] at(0,1pt){};\end{tikzpicture}) in Figure \ref{fig:fgo-arch} and represents the relationship between the states of two consecutive time instants. The core idea is to impose the evolution of sensorized nodes, as we have access to all the measurements within the considered time window. This applies to both head and demand states, although the way of handling non-sensorized nodes differs between them. We define the difference between consecutive head states as
\begin{equation}\label{eq:temp_head_factor}
    \begin{split}
        \delta\bm h^{[t,t+1]}=\hat{\bm{h}}^{[t+1]}-\hat{\bm{h}}^{[t]}=\bm P_s\left(\bm h_s^{[t+1]}-\bm h_s^{[t]}\right),\\
        \delta\bm h^{[t,t+1]}_s=\bm h_s^{[t+1]}-\bm h_s^{[t]},
    \end{split}
\end{equation}
\noindent where $\delta\bm h^{[t,t+1]}$ is the difference between consecutive head states, composed in two stages: first the complete vector is defined as the difference between the interpolated states (using FGSI through $\bm P_s$ and considering Remark \ref{remark:FGSI}), and then the sensorized states are corrected using the measurements at the two consecutive time instants. In this way, the temporal evolution is given by the difference in measurements, and the interpolation is used to provide a better estimation of the difference of state in unmetered nodes. For WDNs, the actual hydraulic head vector at a single time instant depends only on the demand pattern and the boundary conditions, so there is no explicit connection between the hydraulic head vectors of two consecutive time instants, except for the measured evolution.

For the demand evolution, the factor can be expressed as:
\begin{align}\label{eq:temp_demand_factor}
        \delta\bm d^{[t,t+1]}=\bm 0_n, &&
        \delta\bm d^{[t,t+1]}_s=\bm d_s^{[t+1]}-\bm d_s^{[t]},
\end{align}
\noindent where $\delta\bm d^{[t,t+1]}$ is the difference between consecutive demand states, initialized with zeros and then filled with measurement differences at the indices of the consumption-metered nodes.% because \textcolor{red}{we do not have an interpolation mechanism for demands}.

\subsubsection{Structural evolution}\label{subsubsec:struc_factor}

Apart from the information provided by the relation between consecutive time instants, we can establish a relationship between the heads of neighboring nodes at each individual instant. This is performed through the structural evolution factor, indicated by a green square (\begin{tikzpicture}[baseline=-2pt]\node[fhfixed, scale=.8] at(0,1pt){};\end{tikzpicture}) in Figure \ref{fig:fgo-arch}. In this way, we explicitly incorporate interpolation into the estimation process of each single head state. To this end, we set a factor with the following expression:
\begin{equation}\label{eq:struc_factor_1}
    \bm h^{[t]}_{u} =\bm S_u\hat{\bm{h}}^{[t]}= \bm S_u \bm P_s \bm h_s^{[t]},
\end{equation}
\noindent where $\bm S_u\in\mathbb{R}^{(n-n_s)\times n}$ is the matrix that encodes the unmetered nodes. In this way, we relate the unknown head states at a certain time instant with the interpolated values from the measured heads at that specific instant.

\subsubsection{Demand measurement integration}\label{subsubsec:demand_measurements}

In the same spirit of the head interpolation integration, demand measurements can be explicitly incorporated through a factor, to give not only the demand difference between consecutive time instants but also the actual value at each time instant. This factor is indicated through a gray square (\begin{tikzpicture}[baseline=-2pt]\node[ffixed, scale=.8] at(0,1pt){};\end{tikzpicture}) in Figure \ref{fig:fgo-arch}. To this end, the following expression is considered:
\begin{equation}\label{eq:demand_measurement_factor_1}
    \bm d^{[t]}_{s} =\bm S_d\bm{d}^{[t]},
\end{equation}
\noindent where $\bm S_d\in\mathbb{R}^{n_d\times n}$ is the matrix that encodes the nodes with demand sensors, with $n_d$ being its amount. 

\begin{remark}
    Therefore, in our notation, $S_p, S_d, S_u$ represents node selection, corresponding to the sets of pressure ($S_p$) and demand ($S_d$) sensors, as well as unknown ($S_u$) nodal pressures. These matrices can be obtained from the identity matrix by selecting the corresponding rows. \eor
\end{remark}

\subsubsection{Zero-sum of demands}\label{subsubsec:zero_sum_demands}

The assumption of the leak being modeled as an extra demand in the leaky node allows us to define a "constraint" with the aim of improving the demand estimation. Specifically, the characteristics of water distribution networks ensure that the sum of the water entering the WDN through the reservoirs must be equal to the consumed (plus lost) water. In this way, a zero-sum demand factor, indicated with an orange square (\begin{tikzpicture}[baseline=-2pt]\node[fzero, scale=.8] at(0,1pt){};\end{tikzpicture}) in Figure \ref{fig:fgo-arch}, can be defined as:
\begin{equation}\label{eq:zero_sum_demand}
     \bm{1}_n^{\top} \bm d^{[t]}_{s} = 0,
\end{equation}
where $\bm{1}_n$ is a column vector of ones, with size $n$.

% \textcolor{red}{CHECK COMMENT}

\subsubsection{Demand-head relation}\label{subsubsec:demand_head}

As previously explained in Section \ref{sec:prelim_leak}, the relationship between the heads of neighboring nodes, and the flow conveyed through the pipe connecting them, is given by an empirically-based expression, such as the Hazen-Williams equation, posed in \eqref{eq:flow-head}. Considering the notation defined in this section, this equation can be expressed in matricial form as 
    $\bm{B}^{[t]}\bm{h}^{[t]}=\left(\bm{T}\bm{q}^{[t]}_h\right)^{1.852}$,
where $\bm{B}^{[t]}\in\mathbb{R}^{m\times n}$ is the analytical incidence matrix, calculated from the actual sign of head differences of neighboring nodes in $\bm{h}^{[t]}$ (so that $\bm{B}^{[t]}\bm{h}^{[t]}\geq 0$), $\bm{q}_h^{[t]}\in\mathbb{R}^m$ is the head-based flow vector, i.e. the flow vector that can be derived from the head states; and $\bm{T}\in\mathbb{R}^{m\times m}$ is the resistance coefficient diagonal matrix, with the $k$-$th$ diagonal element being $\bm{\tau}_k$, defined in Section \ref{sec:prelim_leak}. 
By manipulating the Hazen-Williams matrix form,
the head-based flows can be obtained as:
\begin{equation}\label{eq:HW2}
    \bm{q}^{[t]}_h = \left(\bm{T}^{-1}\bm{B}^{[t]}\bm{h}^{[t]}\right)^{\frac{1}{1.852}}.
\end{equation} 
In addition, the relationship between the demand of a node and the flows of the pipes connected to it is given by the mass conservation equation:
\begin{equation}\label{eq:mass_conserv}
    \bm{d}^{[t]}_h = -\left(\bm{B}^{[t]}\right)^{\top}\bm{q}^{[t]}_h,
\end{equation} 
\noindent where $\bm{d}^{[t]}_h\in\mathbb{R}^n$ is the head-based demand vector%, $\bm{B}^{[k]}_d\in\mathbb{R}^{m\times n_d}$ is a submatrix of $\bm{B}^{[k]}$ with only the columns corresponding to the nodes with demand sensors
\footnote{The minus sign is required by demands being considered as outflows to the node and the selected sign convention within the definition of $\bm{B}^{[t]}$.}.
Thus, combining \eqref{eq:HW2} and \eqref{eq:mass_conserv}, the expression relating head-based nodal demand and heads is
    $\bm{d}^{[t]}_h = -\left(\bm{B}^{[t]}\right)^{\top}\left(\bm{T}^{-1}\bm{B}^{[t]}\bm{h}^{[t]}\right)^{\frac{1}{1.852}}$.

In order to implement non-linear factors, we need to define an error function and its Jacobian with respect to the involved states. In the case of the demand-head relation, indicated in Figure \ref{fig:fgo-arch} with blue squares (\begin{tikzpicture}[baseline=-2pt]\node[fdstate, scale=.8] at(0,1pt){};\end{tikzpicture}), the error function is
\begin{equation}\label{eq:err_demand-head}
    \bm{e}_d^{[t]} = \bm{d}^{[t]} +\left(\bm{B}^{[t]}\right)^{\top}\left(\bm{T}^{-1}\bm{B}^{[t]}\bm{h}^{[t]}\right)^{\frac{1}{1.852}},
\end{equation} 
where $\bm{d}^{[t]}$ is the actual demand state, computed by FGO. The Jacobian of \eqref{eq:err_demand-head} with respect to $\bm{d}^{[t]}$ is $\bm{J}_d^{[t]}=\bm{I}_n$, while the Jacobian of \eqref{eq:err_demand-head} with respect to $\bm{h}^{[t]}$, i.e., $\bm{J}_h^{[t]}$, is composed of the set elements of $\iota_{i,j}^{[t]}$:
%-------------
% \begin{equation}
% \begin{aligned}
\begin{multline}
    \iota_{i,j}^{[t]} = \biggl\{
    v_{ik}^{[t]}
      \left(
        \frac{1}{1.852}
        \left(
          \frac{1}{\tau_{kk}}
          \bm{b}_{k}^{[t]}
          \bm{h}^{[t]}
        \right)^{\frac{1}{1.852}-1}
      \right)\frac{b_{kj}}{\tau_{kk}}
    \ \biggr|\ \\
    \forall k \text{ s.t. } 
      \exists e_k = e_{ij} \in \mathcal{E}
      %u \in \mathcal{N}_j
    \biggr\},
    \end{multline}

% \end{aligned}
% \end{equation}
%-------------
\noindent where
$i$ iterates the demands,
$j$ iterates the heads,
and $k$ the pipes connecting to the nodes $i$ and $j$.
Here we denoted
$\bm{V}^{[t]}=\left(\bm{B}^{[t]}\right)^{\top}$,
$\tau_{kk}$ is the $k$-$th$ element of the diagonal of $\bm{T}$
and
$\bm{b}_{k}^{[t]}$ is the $k$-$th$ row of $\bm{B}^{[t]}$.
The head Jacobian $\bm{J}_h^{[t]}$ becomes:
%-------------
\begin{equation}\label{eq:Jacobian}
[J_h]^{[t]}_{ij}= 
    \begin{cases}
        \iota_{i,j}^{[t]},& i\neq j,  \\ 
        \sum_{u=1}^n \iota_{u,j}^{[t]},& i= j, \\
        0, & \mbox{otherwise}.
    \end{cases}
\end{equation}
%-------------
% \begin{equation}\label{eq:Jacobian}
% [J_h]^{[t]}_{ij}= \sum_{\iota \in \iota_{i,j}^{[t]}} \iota
% \end{equation}
% %-------------
% \begin{equation}
% \begin{split}
%     [J_h]^{[t]}_{ij}=
%     \begin{cases}
%         v_{ik}^{[t]}\left(\frac{1}{\nu}\left(\frac{1}{\tau_{kk}}\bm{b}_{k}^{[t]}\bm{h}^{[t]}\right)^{\frac{1}{\nu}-1}\right)\frac{b_{kj}}{\tau_{kk}},& (I)  \\ 
%         \sum_{u=1}^n v_{io}^{[t]}\left(\frac{1}{\nu}\left(\frac{1}{\tau_{oo}}\bm{b}_{o}^{[t]}\bm{h}^{[t]}\right)^{\frac{1}{\nu}-1}\right)\frac{b_{oj}}{\tau_{oo}},& (II) \\
%         0, & (III)
%     \end{cases} \\
%     (I)\;\;i\neq j, \mathscr{e}_{k}\in\mathcal{E}; (II)\;\;i= j, \mathscr{e}_{o}\in\mathcal{E}; (III)\;\;\mbox{otherwise}
%     \end{split}
% \end{equation}

\subsubsection{Pressure residual computation}\label{subsubsec:residual_comp}

In order to include localization within the FGO approach, a localization state is defined as $\bm l^{[t]}$. This state represents the pressure residuals, that is, the difference between the leak and leak-free head vectors. This metric is usually used as an indicator of the presence of a leak \cite{Perez2014}. In order to compute the pressure residuals at each time instant, we define the following function and integrate it into the associated yellow (\begin{tikzpicture}[baseline=-2pt]\node[fbar, scale=.8] at(0,1pt){};\end{tikzpicture}) factors from Figure \ref{fig:fgo-arch}
\begin{equation}\label{eq:loc_factor_1}
        \bm{l}^{[t]}=\bm{h}^{[t]}-\overline{\bm{h}}^{[t]},
\end{equation}
where $\overline{\bm{h}}^{[t]}$ is the leak-free head state at time $t$. This leak-free state vector is estimated beforehand (through a previous FGO run), and its used within the localization FGO as a fixed value.%, which needs to be computed previously to the leak-related FGO computation, along with all the other leak-free states for the remaining time instants. 

\subsubsection{Leak localization constraint}\label{subsubsec:leak_loc_const}

A final factor is used to relate the state of consecutive localization states in time, through the assumption that pressure residuals should mostly be similar if the boundary conditions are similar, as the leak is considered to be constant along the scenario (the leak size may change, but not the location). Therefore, we impose through a factor (indicated with purple squares \begin{tikzpicture}[baseline=-2pt]\node[flstate, scale=.8] at(0,1pt){};\end{tikzpicture} in Figure \ref{fig:fgo-arch}) that consecutive localization states must be as similar as possible:
\begin{equation}\label{eq:loc_factor_2}
        \bm{l}^{[t+1]}=\bm{l}^{[t]}.
\end{equation}
The FGO output is used to obtain the final localization metric
\begin{equation}\label{eq:norm_loc}
    \tilde{\bm{l}}_{g} = -\frac{\bm{l}^{[0]}-\mbox{min}\left(\bm{l}^{[0]}\right)}{\mbox{max}\left(\bm{l}^{[0]}\right)-\mbox{min}\left(\bm{l}^{[0]}\right)},
\end{equation}

\noindent considering that we select $\bm{l}^{[0]}$ because the localization states are forced to be almost equal, and the negative sign is used in to impose that the node with the highest leak likelihood is associated with the highest residual value.

% \todo{Luis, Paul: describe all the water network specific factor nodes in our architecture: Jacobian, in between factors}

\subsection{Estimation and localization algorithm}\label{sec:method_estimation}
%--------------
\begin{figure*}
    \centering
    \includegraphics[width=\linewidth]{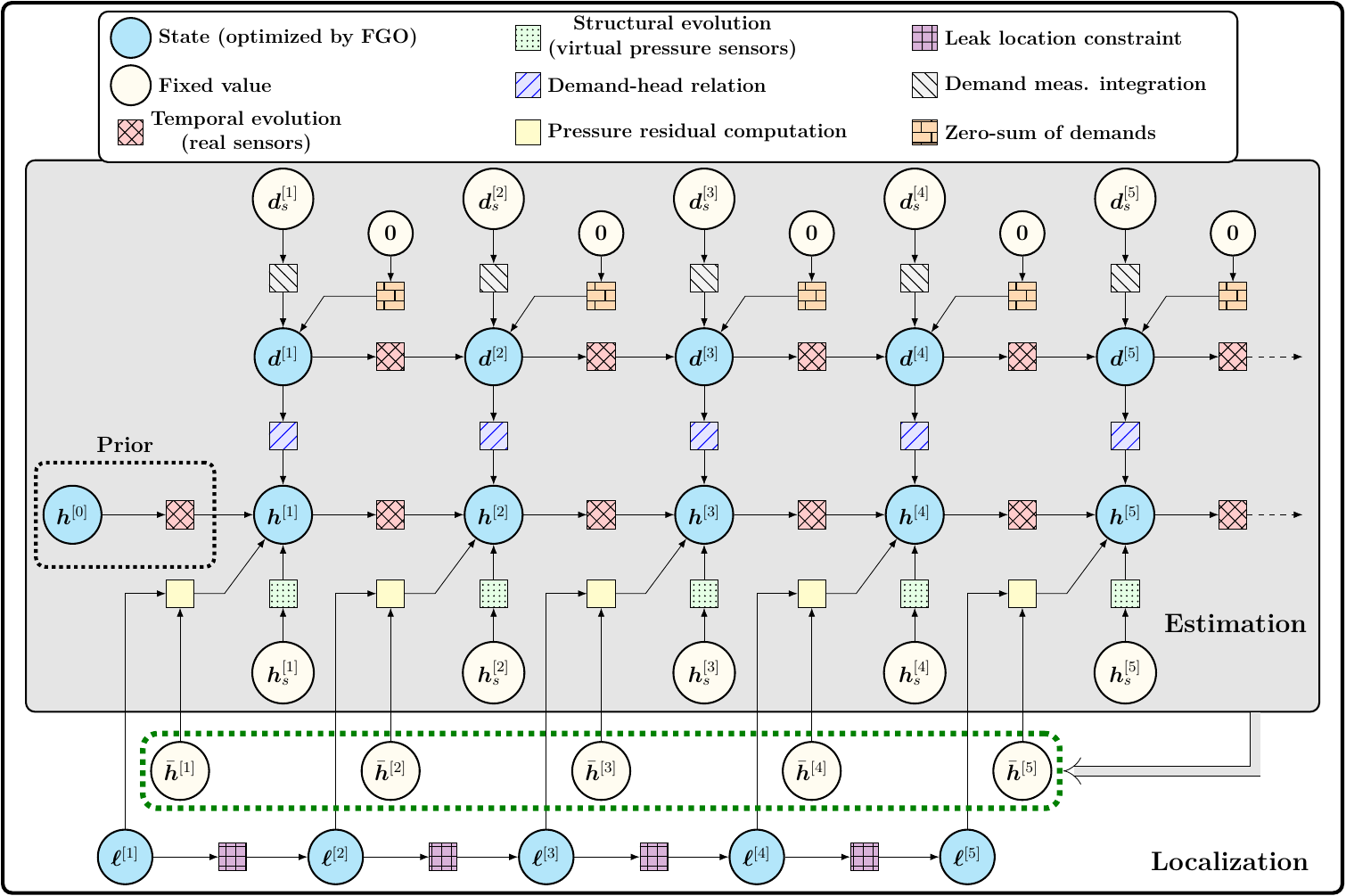}
    \caption{Estimation and localization factor graphs
    for five time instances corresponding to five sensor readings
    starting from prior $\bm{h}^{[0]}$.
    The arrows depict the passage of time.
    The gray area presents the estimation process of the water network behavior in a leak-free scenario producing the head values $\bm{\overline{h}}^{[t]}$ grouped inside the green rectangle.
    The white area presents the specific localization factor graph elements
    which include the leak-free head values
    and the new localization states $\bm{l}^{[t]}$.
    These will be used in the optimization process
    to estimate the network heads $\bm{h}^{[t]}$
    and the localization states
    in order to obtain the leaky network nodes.}
    \label{fig:fgo-arch}
\end{figure*}
%--------------

%--------------
\begin{algorithm}
\caption{Factor Graph Leak Localization (FGLL)}\label{alg:fgo}
\begin{algorithmic}[1]
% \COMMENT Build FGO for $N$ sensor readings\\
\STATE Add prior temporal evolution factor $\bm{h}^{[0]} \to \bm{h}^{[1]}$
\FOR{$t \gets 1$ to $T$}
\STATE Add temporal evolution factor $\bm{h}^{[t]} \to \bm{h}^{[t+1]}$ and $\bm{d}^{[t]} \to \bm{d}^{[t+1]}$\label{lst:line:temporal}
\STATE Add structural evolution factor $\bm{h}_s^{[t]} \to \bm{h}^{[t]}$
\STATE Add demand measurement factor $\bm{d}_s^{[t]} \to \bm{d}^{[t]}$
\STATE Add zero-sum demand constraint factor for $\bm{d}^{[t]}$
\STATE Add demand-head factor $\bm{d}^{[t]} \to \bm{h}^{[t]}$\label{lst:line:demand2head}
\ENDFOR
% \COMMENT Estimation
% \FOR{$t \gets 1$ to $T$}
\STATE Insert pressure reading $\bm{h}_s^{[t]}$  for time instances $t \in [T]$\label{lst:line:pressure}
\STATE Insert demand reading $\bm{d}_s^{[t]}$ for time instances $t\in[T]$\label{lst:line:demand}
% \ENDFOR
\STATE Perform leak-free estimation FGO: obtain $\overline{\bm{h}}^{[t]}, \; \forall t \in [T]$ %and $\overline{\bm{d}}^{[t]}$
% \COMMENT Localization
\FOR{$t \gets 1$ to $T$}
\STATE Add factors as in lines L\ref{lst:line:temporal}--L\ref{lst:line:demand2head}
\STATE Add pressure residual factor $(\bm{l}^{[t]}, \overline{\bm{h}}^{[t]}) \to \bm{h}^{[t]}$
\STATE Add leak localization constraint $\bm{l}^{[t]} \to \bm{l}^{[t+1]}$
\ENDFOR
\STATE Insert current sensors readings as in lines L\ref{lst:line:pressure} -- L\ref{lst:line:demand}
\STATE Perform localization FGO: obtain leak states $\bm{l}^{[t]}, \forall t\in[T]$
\STATE Compute $\tilde{\bm{l}}_{g}$ using \eqref{eq:norm_loc}
\end{algorithmic}
\end{algorithm}
%--------------

Algorithm~\ref{alg:fgo} describes the entire process of 
leak-free state estimation
and leak localization as depicted 
in Figure \ref{fig:fgo-arch}.

Estimation
starts from the prior state $\bm{h}^{[0]}$
which can be given by the user or obtained through marginalization
from the previous time window (line 1).
Lines 2 to 8 build the leak-free estimation factor graph
depicted by the gray area in Figure~\ref{fig:fgo-arch}.
Line 3 adds the temporal evolution of the state
for heads~\eqref{eq:temp_head_factor}
and demands~\eqref{eq:temp_demand_factor}.
The demands state evolution acts also as a soft switchable constraint:
while the end goal is to include demand information in the head state (line 7),
this temporal demand evolution ensures that our demand approximation is consistent between time instances.
Next,
line 4 performs the interpolation through
head state structural evolution~\eqref{eq:struc_factor_1}
and line 5 performs the demand measurement factor~\eqref{eq:demand_measurement_factor_1}.
Line 6 adds the zero-sum demands constraint factor~\eqref{eq:zero_sum_demand}
and line 7 adds the Hazen-Williams based non-linear factor which translates demands to be incorporated in the head state~\eqref{eq:err_demand-head}.
The factor graph is now complete and can be used for time windows of length $T$ by inserting the sensor readings into leaf nodes $h_s$ and $d_s$ (lines 9-10)
and performing factor graph optimization (line 11).
The optimization process provides leak-free head states $\overline{\bm{h}}$
for all water network nodes (including the non-sensorized nodes)
across all time instances.
These will be used in the localization process.

Leak localization is based on the factor graph and the leak-free states obtained in the estimation process.
Lines 12--16 build the localization factor graph
depicted by the white area in Figure~\ref{fig:fgo-arch}.
Line 13 reconstructs the estimation factor graph including the prior $\bm{h}^{[0]}$.
In line 14, the localization factor combines the leak-free head state,
obtained in the FGO estimation from line 11,
with the leak localization state~\eqref{eq:loc_factor_1}.
Finally, line 15 adds the leak localization constraint factor~\eqref{eq:loc_factor_2}
which acts as a switchable constraint meant to maintain consistency across time instances such that a detected leak is propagated across states.
The factor graph is now complete, and the sensor readings
are inserted in line 17 as in the leak-free estimation case. %, withthe main difference is that now we are running under a \mbox{(multiple-)leak} scenario.
The optimization of factor graph of line 18
provides us with the leak states $\bm{l}^{[t]}$ in all time instants
that will be used in line 19 to compute the localization metric described around \eqref{eq:norm_loc}. This metric is designed such that the most likely leak candidates are assigned values close to 1, whereas the unlikely candidates have a value close to zero. The localization result is obtained by means of a statistical analysis of the localization metric. Specifically, the mean and standard deviation of this vector are computed, only accepting as candidates those nodes whose localization metric value is higher or equal to the sum of those two statistics. 

Our algorithm complexity is bound by the FGO solver used in lines 11 and 18.
In our implementation, we use the incremental smoothing and mapping (iSAM2) solver~\cite{Kaess12_isam2}
which has a worst-case scenario bound of $O(\abs{\mathcal V}^3)$,
where $\abs{\mathcal V}$ is the number of factor graph variable nodes.
In practice, the authors show that the complexity is much lower.
% proving a $O(n^{1.5})$ bound for sparse systems.

% \todo{Paul, Luis: describe the estimation factor graph architecture with switchable constraints and everything else}

% \subsection{Localization}\label{sec:method_estimation}

%The output of Algorithm~\ref{alg:fgo} is the localization metric $\tilde{\bm{l}}_{g}$. 

%\todo{Luis, Paul: describe how the estimation FGO is static now and we use it in the new FGO localization architecture to perform leak localization}

%------------------------------------------------------------------
\section{Case Studies}\label{sec:studies}

% \todo{Luis}

In order to assess the performance capabilities of the proposed methodology, a set of networks will be considered as benchmarks, allowing comparison with established strategies from the literature. These networks are selected to cover a broad range of characteristics, including varying scale (number of nodes/pipes), demand pattern complexity, sensor deployment, levels of uncertainty, etc. Their main properties are summarized in Table \ref{table:network-data}. To highlight that the water inflow to each scenario is expressed through the mean and standard deviation of the inflow for the considered time window, except for L-TOWN, where we consider the whole year 2018 to account for its variability\footnote{As explained in subsequent sections, we consider only Area A, and therefore the inflow is computed to account up for this area only.}. In terms of the sensorization of  networks, water inlet nodes are always considered to be sensorized both in pressure and demand.

\begin{table}
\centering
\caption{Summary of the main properties of the selected benchmarks}
\label{table:network-data}
\begin{tabular}{l|c|c|c}
% \hline
 & \textbf{T-Example} %& \textbf{Hanoi}
 & \textbf{Modena} & \textbf{L-TOWN-A} \\
\hline
inflow ($\ell/s$) & 88.47$\pm$27.62 %& 
& 399.51$\pm$119.22  & 164.62$\pm$55.55 \\
% \hline
\# junctions & 9 %& 31 
& 272 & 657 \\
% \hline
\# pipes & 10 %& 34 
& 317 & 766 \\
% \hline
\# water inlets & 1 %& 1 
& 4 & 2 \\
% \hline
\# press. sensors & 4 %&- 
& 20 & 31 \\
% \hline
\# demand sensors & 4 %&- 
& 40 & 31 \\
% \hline
\end{tabular}
\end{table}

% \subsection{Synthetic Example}\label{sec:study_synthetic}
\textbf{Synthetic Example.} The first assessment of the proposed methodology is performed on an artificial benchmark, originally created to serve as a simple network example to be able to quickly prototype and test. The properties of this network, called the T-example, are presented in Table \ref{table:network-data}, and a graphical representation is shown in Figure \ref{fig:t-example}. Four sensors have been placed on the network to provide nodal head and demand data. These sensors provide pressure and demand information every 5 minutes. 

\begin{figure}
    \centering
    \includegraphics[width=\linewidth]{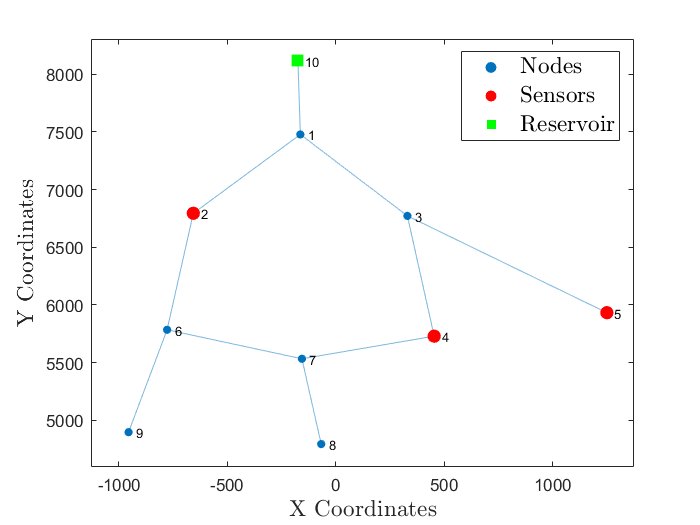}
    \caption{Graphical representation of the synthetic network known as T-example.}
    \label{fig:t-example}
\end{figure}

% \begin{table}[t]
% \caption{Weighted averaged pipe distance error, comparing GSI, AW-GSI, UKF-AW-GSI and FGO.}
% \label{table:times}
% \footnotesize
% \centering
% \begin{tabular}{c|c|c|c|c}
% %\hline
% \textbf{Leak ID} & \textbf{GSI} & \textbf{AW-GSI} & \textbf{UKF-AW-GSI} & \textbf{FGO} \\
% \hline
% \textit{1} &  &  &  &  \\
% \textit{2} &  &  &  &  \\
% \textit{3} &  &  &  &  \\
% \textit{4} 2.15 & 2.07 & 2.04 &  \\
% \textit{5} &  &  &  &  \\
% \textit{6} &  &  &  &  \\
% \textit{7} &  &  &  &  \\
% \textit{8} &  &  &  &  \\
% \textit{9} &  &  &  &  \\
% %\hline
% \end{tabular}
% \end{table}

% \begin{table}[t]
% \caption{Normalized leak metric for each potential leak, comparing GSI, AW-GSI, UKF-AW-GSI and FGO.}
% \label{table:times}
% \footnotesize
% \centering
% \begin{tabular}{c|c|c|c|c|c|c|c|c|c|c}
% \hline
% \textbf{Node} & 1 & 2 & 3 & 4 & 5 & 6 & 7 & 8 & 9 & R \\
% \hline
% \multicolumn{11}{c}{\textbf{Leak 4}} \\
% \hline
% \textbf{GSI} & 0.40 & 1.00 & 0.45 & 0.00 & 0.64 & 0.34 & 0.34 & 0.37 & 0.32 & 0.04 \\
% \textbf{AW-GSI} & 0.54 & 1.00 & 0.62 & 0.11 & 0.62 & 0.41 & 0.30 & 0.36 & 0.46 & 0.00 \\
% \textbf{(III)} & 0.82 & 0.18 & 0.76 & 0.00 & 0.44 & 0.99 & 0.90 & 1.00 & 1.00 & 0.77 \\
% \textbf{FGO} &  0.12 & 0.41 & 0.30 &  1.00 & 0.20 & 0.86 & 0.93 & 0.94 & 0.87 & 0.00 \\
% \hline
% %\hline
% \end{tabular}
% \end{table}

% \subsection{Hanoi Network}\label{sec:study_hanoi}

% \subsection{Modena Network}\label{sec:study_modena}
\textbf{Modena Network.} The water network of Modena, Italy, is a well-established benchmark in the development of leak management methods \cite{Bragalli2012, Alves2022, Irofti2023}. The network scheme is represented in Figure \ref{fig:Modena}, while its characteristics are summarized in Table \ref{table:network-data}. The scale of the network, measured in terms of number of nodes/pipes and total consumption, reflects the complexity of a problem of realistic dimension, comparable to that of real-world water distribution systems. Additionally, the benchmark is open-source, promoting transparency and reproducibility. A total of 20 pressure sensors and 40 demand sensors (20 of which coincide with the pressure sensors), provide measurements with a sample time of 1 hour.

\begin{figure}
    \centering
    \includegraphics[width=\linewidth]{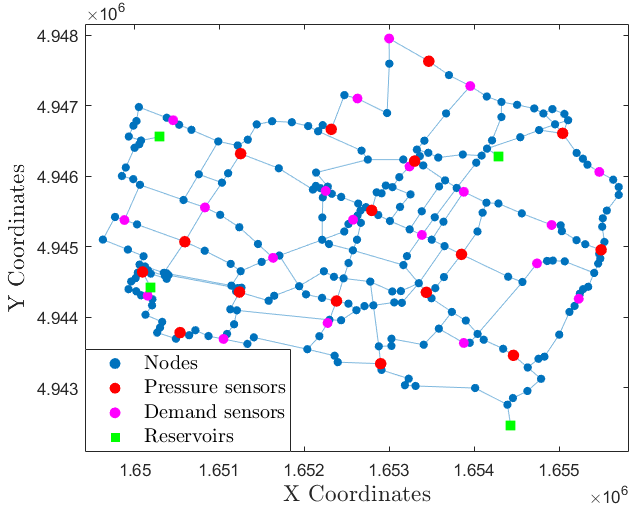}
    \caption{Graphical representation of the Modena network.}
    \label{fig:Modena}
\end{figure}

% \subsection{L-TOWN Network}\label{sec:study_ltown}
\textbf{L-TOWN Network.} L-TOWN was introduced in the Battle of Leakage Detection and Isolation Methods 2020 - BattLeDIM2020 \cite{Vrachimis2022}, organized during the 2nd International CCWI/WDSA Joint Conference. The challenge consisted of the evaluation and comparison of leak detection and localization schemes, fed with SCADA measurements from strategically placed pressure, demand, and flow sensors (as well as a hydraulic model for model-based methods).

The L-TOWN network is composed of 782 junctions, 2 reservoirs, 1 tank, 4 valves, and 905 pipes, with total pipe length of 42.6 km. The WDN is divided into three areas, depending on the nodal elevation.
\textit{Area A} is made up of 657 nodes, at elevations between 16 and 48 meters. Two water inlets feed the WDN through this zone, which then feeds Area B (connected through a Pressure Regulating Valve (PRV)) and Area C (fed through a tank filled from Area A by means of a pump). In addition, there are 29 pressure sensors and 3 flow meters.
\textit{Area B} comprises 31 nodes, all below 16 meters. In particular, Area B has only one pressure meter.
\textit{Area C} has 94 nodes are above 48 meters of elevation. There are 3 pressure sensors and 82 Automated Meter Reading (AMR) sensors, which provide demand data.

Moreover, the competition organizers provided two datasets to the participant teams: (i) 2018 dataset, which served as a training and calibration set; and (ii) 2019 dataset, which was used to evaluate the performance of the methods. However, in this work, we have evaluated the methodologies using the 2018 dataset, since the extreme multileak conditions of the 2019 dataset could hinder drawing meaningful conclusions from the results. Specifically, we will focus on Area A, as most of the leaks in 2018 occurred in this zone. Therefore, we are going to address the challenge posed by the network in Figure \ref{fig:LTOWN}, which has the properties indicated in Table \ref{table:network-data}.

\begin{figure}
    \centering
    \includegraphics[width=\linewidth]{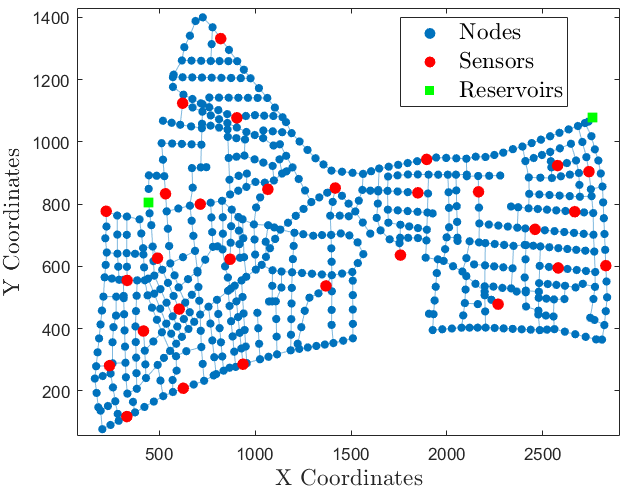}
    \caption{Graphical representation of the Area A of the L-TOWN network.}
    \label{fig:LTOWN}
\end{figure}

%------------------------------------------------------------------
\section{Results}\label{sec:sim}

% \todo{Luis, Paul}

In order to illustrate the capabilities of the proposed methodology, this section presents a performance comparison with respect to other three data-driven state estimation methods, which are associated to a localization stage to additionally analyze the leak localization performance. In this way, both estimation and localization results are presented for each methodology in each case study, leading to a comprehensive analysis of the performance of the proposed method.

Specifically, the comparison estimation approaches are:
% \begin{enumerate}[label=(\Roman*), left=0pt, leftmargin=2.5em]
    (I)~GSI \cite{RomeroBen2022}, implemented as explained in Section \ref{sec:interpolation};
    % \item AW-GSI-LCSM \cite{Irofti2023}. This approach improves GSI-LCSM to improve the graph weights, by introducing a novel physical-weighting scheme. 
    (II)~UKF-AW-GSI \cite{RomeroBen2024b} which uses a physically-informed version of GSI, denoted as AW-GSI \cite{Irofti2023}, and integrates demand measurements to current used pressure data by using a UKF-based strategy; %This state estimation enhancement method is applied to both the leak and leak-free data.
    (III)~GHR-S \cite{Zhou2023} which exploits graph signal processing theory using slow varying parts of the nodal heads and pseudo-measurements to improve the estimation process stability. 
% \end{enumerate}

The localization stage will be performed using an approach denoted as Leak Candidate Selection Method (LCSM) \cite{RomeroBen2022}. This strategy considers the leak-free and leak interpolated states as the x-y coordinates of a cloud of 2-D points, and considers the nodes associated to furthest points to the best-fitting line to the cloud as the most likely leak candidates. To this end, it also computes the mean and standard deviation of this vector of distances and uses their sum as a threshold. 

General parameters of the comparison methods must be settled, such as the directionality weight of GSI, i.e. $\chi=10^3$; the UKF standard parameters, i.e., $\alpha=10^{-3}$, $\beta=2$ and $\kappa=0$. Regarding FGLL, the noise covariance matrices for the different factors are computed as a scalar value multiplied by the identity matrix (whose size depends on the factor and benchmark). These scalars are indicated in Table~\ref{table:fgll_noise}.
Our implementation is available at \url{https://github.com/pirofti/FGLL}.

\begin{table}
\centering
\caption{Noise covariance levels for the different factors in FGLL, for the bechmarks of T-Example, Modena and L-TOWN-A.}
\label{table:fgll_noise}
\begin{tabular}{l|c|c|c}
% \hline
\textbf{Factor} & \textbf{T-Example} %& \textbf{Hanoi}
 & \textbf{Modena} & \textbf{L-TOWN-A} \\
\hline
\begin{tikzpicture}[baseline=-2pt]\node[fstate, scale=1] at(0,1pt){};\end{tikzpicture} Temporal
 & $10^{-12}$ & $10^{-12}$ & $10^{-12}$ \\
% \hline
\begin{tikzpicture}[baseline=-2pt]\node[fhfixed, scale=1] at(0,1pt){};\end{tikzpicture} Structural 
 & $10^{-4}$ & $10^{-4}$ & $10^{-4}$ \\
% \hline
\begin{tikzpicture}[baseline=-2pt]\node[fdstate, scale=1] at(0,1pt){};\end{tikzpicture} Demand-head
 & $10^{-12}$ & $10^{-12}$ & $10^{-12}$ \\
% \hline
\begin{tikzpicture}[baseline=-2pt]\node[fbar, scale=1] at(0,1pt){};\end{tikzpicture} Pressure residual
 & $10^{-3}$ & $10^{-3}$ & $10^{-3}$ \\
% \hline
\begin{tikzpicture}[baseline=-2pt]\node[flstate, scale=1] at(0,1pt){};\end{tikzpicture} Leak localization
 & $10^{-5}$ & $10^{-9}$ & $10^{-5}$ \\
% \hline
\begin{tikzpicture}[baseline=-2pt]\node[ffixed, scale=1] at(0,1pt){};\end{tikzpicture} Demand meas.
 & $10^{-4}$ & $10^{-4}$ & $10^{-4}$ \\
% \hline
\begin{tikzpicture}[baseline=-2pt]\node[fzero, scale=1] at(0,1pt){};\end{tikzpicture} Zero-sum demands
 & $10^{-12}$ & $10^{-9}$ & $10^{-12}$ \\
 % \hline
\end{tabular}
\end{table}

\subsection{T-example}

The evaluation of the performance of the methods considered over T-example requires the generation of leak-free and leaky data. To this end, the hydraulic simulation tool EPANET \cite{Rossman2000} is used to simulate the corresponding scenarios, including a leak-free event and 9 leak scenarios, one per potential leak.
In order to maintain the simplicity of the synthetic scenario, a constant leak size of 0.5 $\ell/s$ is considered (representing a $\sim$0.57\% of the average inflow), and noise is not considered. Each simulation has a hydraulic time step of 5 minutes, and they last for one day, leading to 288 samples per scenario. 

Regarding the specific parameterizations for this network, the process and measurement noise covariance matrices are $Q=I_n$ and $R=10^{-4}I_{2n_s}$ respectively, where $n_s$ is the number of pressure/demand sensors, and the number of UKF internal iterations, $k_{loop}$, is settled to 200. Moreover, the GHR-S parameters for the frequency limit of the first stage, $f_p$, and the limit for the final stage, $f$, are set to $f_p = 3$ and $f=7$. 

\subsubsection{Estimation}

The estimation results are presented in Table \ref{table:estimation_Texample}. The error is computed through the root mean square error,
$RMSE(\bm x, \hat{\bm x}) = \sqrt{\frac{1}{n} \sum_{i=1}^{n} (x_i - \hat{x}_i)^2}$, where $\bm{x}$ is a generic vector. The RMSE is computed between the estimated and simulated heads at each time instant, and then the average and standard deviation are computed for all the time instants. Note that the three comparison methods operate over individual samples, i.e., they interpolate/estimate the state corresponding to a specific set of measurements, gathered at a single time instant. However, FGLL computes the estimation of all time instants at the same time.

\begin{table}[t]
\caption{RMSE of the head estimation process in T-example, comparing GHR-S, GSI, UKF-AW-GSI and FGLL.}
\label{table:estimation_Texample}
\footnotesize
\centering
\begin{tabular}{c|c|c|c|c}
\multicolumn{1}{c|}{\multirow{2}{*}{Leak ID}} & \multicolumn{4}{c}{\textbf{RMSE: $\bm{\mu\pm\sigma}$ (m)}} \\
% \cline{2-5}
\multicolumn{1}{c|}{} & \textbf{GHR-S}\rule{0pt}{2ex} & \textbf{GSI} & \textbf{UKF-AW-GSI} & \textbf{FGLL} \\
\hline
\textit{1} \rule{0pt}{2ex} & 1.60$\pm$0.83 & 0.86$\pm$0.45 & 1.09$\pm$0.56 & 0.87$\pm$0.43  \\
\textit{2} & 1.57$\pm$0.82 & 0.86$\pm$0.45 & 1.09$\pm$0.56 & 0.87$\pm$0.43  \\
\textit{3} & 1.62$\pm$0.84 & 0.86$\pm$0.45 & 1.09$\pm$0.56 & 0.87$\pm$0.43  \\
\textit{4} & 1.62$\pm$0.84 & 0.86$\pm$0.45 & 1.09$\pm$0.56 & 0.87$\pm$0.43  \\
\textit{5} & 1.63$\pm$0.84 & 0.86$\pm$0.45 & 1.09$\pm$0.56 & 0.87$\pm$0.43 \\
\textit{6} & 1.64$\pm$0.84 & 0.86$\pm$0.45 & 1.09$\pm$0.56 & 0.87$\pm$0.43 \\
\textit{7} & 1.62$\pm$0.83 & 0.88$\pm$0.45 & 1.12$\pm$0.57 & 0.90$\pm$0.44 \\
\textit{8} & 1.61$\pm$0.83 & 0.88$\pm$0.45 & 1.12$\pm$0.57 & 0.90$\pm$0.44 \\
\textit{9} & 1.68$\pm$0.85 & 0.92$\pm$0.46 & 1.15$\pm$0.58 & 0.93$\pm$0.44 \\
%\hline
\end{tabular}
\end{table}

The results show how GSI and FGLL lead to a rather similar performance, whereas GHR-S and UKF-AW-GSI present a degraded performance in comparison. In the case of UKF-AW-GSI, this is caused by the lack of sensorization in nodes 6-9, because the UKF-based strategy tends to reduce the differences between neighboring states (heads) through the prediction function, which performs graph diffusion. Therefore, the lack of sensorization at these nodes prevents the UKF from correcting this tendency. In terms of computation time, the combined estimation time for all the considered time instants is $\sim$0.02 s, $\sim$1.33 s and $\sim$35.81 s for GHR-S, GSI and UKF-AW-GSI respectively, while FGLL obtains all of them in a single optimization that lasts for $\sim$0.80 s\footnote{Estimation times are used as representatives of the practical computational cost of the methods (for FGLL, the factor optimization includes estimation+localization, as explained in Algorithm \ref{alg:fgo}), as the statistical analysis (or LCSM) is practically identical for all the approaches, as well as negligible in terms of computational cost to the more expensive methods.}.

\subsubsection{Localization}

The leak scenarios previously generated are used to perform localization experiments with the different methodologies. Therefore, simulation of a leak-free reference scenario is also required. The localization results are encoded in the images at Figure \ref{fig:loc_t-example}. Each image corresponds to a different approach and encodes through a colourmap (at the right of the image) the corresponding normalized localization metric for each network node (x-axis) at each leak scenario (y-axis). The localization metric for GHR-S-LCSM, GSI-LCSM and UKF-AW-GSI-LCSM is the result of LCSM, whereas FGLL uses the localization metric and statistical analysis explained in Section \ref{sec:method_estimation}. Therefore, the highest the associated value to each pixel, the highest the leak likelihood associated to that node in that leak event. 

\begin{figure}
    \centering
    \includegraphics[width=\linewidth]{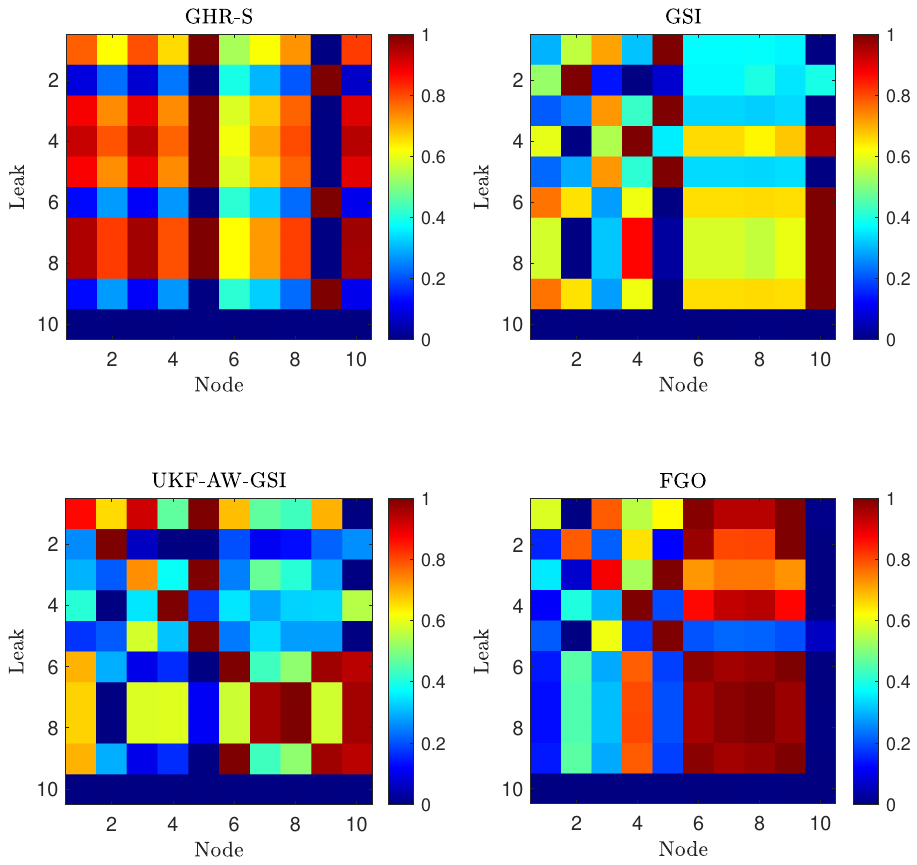}
    \caption{Normalized leak metric for each potential leak, comparing GSI, AW-GSI, UKF-AW-GSI and FGLL. Each image encodes a colour code of the normalized metric of a node (x-axis) in a leak scenario (y-axis).}
    \label{fig:loc_t-example}
\end{figure}

The results show how the localization performance of GHR-S-LCSM is not satisfactory in this case, with nodes 4 and 9 selected as the most likely candidates in most scenarios. 
GSI-LCSM produce a satisfactory result for leaks 2, 4 and 5, which are properly located (the diagonal pixel is among the highest in localization metric). However, most of the rest of leak scenarios show a degraded performance, with the most noticeable case in leaks from 6 to 9, where node 10, that is, the reservoir, is always associated to the highest probability. UKF-AW-GSI-LCSM leads to a satisfactory performance, implying that the state estimation of both leak and leak-free data aligns them in a proper manner. 

Finally, FGLL results present several advantages:
% \begin{itemize}
    The value in the diagonal elements indicates a better correlation between leak scenario and proposed leak location with respect to GSI-LCSM;
    although the results for leaks 6-9 can be regarded as similar in those nodes, it is noticeable how for leaks 6 and 9, the nodes with the highest value are 6 and 9, and the same occurs with leaks 7 and 8. This implies a degree of precision that was not shown by the other methods (except for UKF-AW-GSI-LCSM);
    FGLL discards the reservoir node as likely, leading to a more sound result than UKF-AW-GSI-LCSM.
% \end{itemize}

% \begin{figure}[htb]
%     \centering
%     \includegraphics[width=\linewidth]{figures/texample/leaks_texample.png}
%     \caption{Normalized leak metric for each potential leak, comparing GSI, AW-GSI, UKF-AW-GSI and FGO.}
%     \label{fig:t-example}
% \end{figure}

\subsection{Modena}

Leak and leak-free data must be obtained from hydraulic simulations in EPANET \cite{Rossman2000} and
in the case of Modena we selected 32 leak scenarios, distributed uniformly across the network, to reduce the total computation time. These leak scenarios (as well as the leak-free scenario) include data from 3 days of measurements, leading to 72 time instants. 
Moreover, three leak sizes are considered: 4.5, 5.5 and 6.5 $\ell/s$ (representing a 1.13\%, 1.38\% and 1.63\% of the average inflow, respectively). Finally, a 1\% random noise level is considered in the roughness and diameter values of the pipes (due to the changing nature of these parameters, caused by the age and state of the pipes), and a 0.5\% noise in the demand patterns.

Regarding the specific parameterizations in the Modena case study, the UKF process and measurement noise covariance matrices are $Q=I_n$ and $R=10^{-4}I_{n_s+n_d}$ respectively, where $n_d=40$ is the number of demand sensors, and the number of UKF internal iterations is $k_{loop}=75$, reducing the UKF computational cost. For GHR-S we set $f_p = 6$ and $f=14$ (around a 70\% of the number of pressure sensors \cite{Zhou2023}).

\subsubsection{Estimation}

The estimation results for the Modena case are presented in Table \ref{table:estimation_results_L-TOWN_Modena}. Again, we compute the RMSE between the estimated and simulated heads, averaging among all the time instants. Moreover, the presented results are averaged over the leak scenarios as well, as the number of considered leaks impedes showing the individual results for each single leak scenario. The table presents the performance results for both accuracy and computation time.

With respect to the estimation error, GHR-S leads to the least accurate result. GSI and UKF-AW-GSI are capable of outperforming it, with the latter leading to an excellent level of accuracy. In the case of FGLL, the estimation process is capable of improving the performance of GHR-S, but it is outperformed by the GSI-based methods. This is justified by the design of FGLL which optimizes  in a single run the estimation of all the considered time instants within the time window. This implies that the optimization process would pursue a solution that leads to a generalized reduction of the estimation error. Considering that there is no explicit relationship between the heads of a node in consecutive time instants (as they only depend on the demand patterns and boundary conditions at the specific time instant), the parallel estimation of correlated estimated states harms individual estimation performance.

Regarding the computation time, GHR-S is the faster method, whereas GSI reduces the estimation error by an affordable increase in computation time. The UKF-based strategy requires a drastically longer period to compute the solution, although again an error reduction can be achieved. FGLL is around 15 times faster than UKF-AW-GSI.

\begin{table}[t] % Example float position options
    \caption{Average head estimation process RMSE in (I) Modena and (II) L-TOWN-A, comparing GHR-S, GSI, UKF-AW-GSI and FGLL.}
    \centering
    \resizebox{\linewidth}{!}{
    \begin{tabular}{cc|cccc}
        % \hline
        & \multicolumn{1}{c|}{\multirow{2}{*}{\textbf{Metric}}} & \multicolumn{4}{c}{\textbf{Estimation results ($\mu\pm\sigma$)}} \\
         % \cline{3-6}
        & \multicolumn{1}{c|}{}& \textbf{GHR-S}\rule{0pt}{2.5ex} & \textbf{GSI} & \textbf{UKF-AW-GSI} & \textbf{FGLL}  \\
         \hline
        \multirow{2}{*}{(I)}
        & $RMSE$ (m) &
        $4.41\pm2.06$ & $2.10\pm0.98$ & $1.41\pm0.66$ & $3.29\pm2.26$ \\
        % \hline
        & $t$ (s) &
        $ 0.13\pm0.01 $ & $ 19.62\pm2.72 $ & $ 936.18\pm 105.40$ & $ 60.74\pm 6.99$ \\
        \hline
        \multirow{2}{*}{(II)}
        & $RMSE$ (m) &
        $0.15\pm0.01$ & $0.15\pm0.01$ & $0.06\pm0.01$ & $0.22\pm0.01$ \\
        & $t$ (s) &
        $0.19\pm0.01$ & $1.67\pm0.12$ & $969.77\pm25.96$ & $249.05\pm68.28$ \\
        % \bottomrule
        % \hline
        % \hline
    \end{tabular}%
    }
    % \vspace{3pt}
    \label{table:estimation_results_L-TOWN_Modena}
\end{table}

\subsubsection{Localization}

Again, the leak scenarios considered to obtain the estimation results are used to perform an analysis of the localization performance. The localization results are presented in Table \ref{table:localization_results_L-TOWN_Modena}. The results are presented as the candidate-to-leak distance error, which measures the shortest-path pipe distance (in both kilometers and the number of pipes) between the actual leak and the selected number of top candidates. Here, two different metrics are presented.
\textit{Best} indicates the actual shortest-path distance between the leak and the most likely candidate (based on the corresponding localization metric: LCSM metric for GHR-S, GSI and UKF-AW-GSI, and the statistical analysis of the localization metric in FGLL). This behaves as a representative of the node-level\footnote{Leak localization performance is categorized as ``node-level'',  identifying the exact leak location, or ``area-level'', where a wider search area is provided.} performance of the method.
\textit{Avg. 5} shows the averaged pipe distance between the leak and the top 5 candidates. This serves as a representative of the area-level result of each method.

The presented results show how UKF-AW-GSI-LCSM and FGLL outperform both GSI-LCSM and GHR-S-LCSM (with the former outperforming the latter too). Analyzing the performance of UKF-AW-GSI-LCSM and FGLL, it is noticeable how the results are quite similar, with UKF-AW-GSI-LCSM leading to a more accurate node-level localization, and FGO producing less error in terms of area-level. Incorporating the estimation results to the analysis, we can conclude that FGLL do not require an estimation as precise as the one produced by UKF-AW-GSI in order to obtain a comparable localization. Moreover, regarding the computation times in Table \ref{table:estimation_results_L-TOWN_Modena}, we can conclude that a similar localization performance can be obtained 15 times faster using FGLL.

% \begin{table}[t]
% \caption{Candidate-to-leak distance error for (I) Modena and (II) L-TOWN, comparing GHR-S, GSI, UKF-AW-GSI and FGO.}
% \label{table:localization_Modena_LTOWN}
% \footnotesize
% \centering
% \begin{tabular}{c|c|c|c|c}
% \multicolumn{1}{c|}{\multirow{2}{*}{Metric}} & \multicolumn{4}{c}{\textbf{Method}} \\
% \cline{2-5}
% \multicolumn{1}{c|}{} & \textbf{GHR-S}\rule{0pt}{2ex} & \textbf{GSI} & \textbf{UKF-AW-GSI} & \textbf{FGO} \\
% \hline
% \textit{Best (km)} \rule{0pt}{2ex} & 1.30$\pm$0.68 & 0.93$\pm$0.80 & 0.69$\pm$0.69 & 0.72$\pm$0.51  \\
% \textit{Best (pipes)} & 6.19$\pm$3.70 & 4.38$\pm$3.33 & 3.28$\pm$2.91 & 3.34$\pm$1.99  \\
% \textit{Avg. 5 (km)} & 1.33$\pm$0.75 & 0.93$\pm$0.68 & 0.80$\pm$0.57 & 0.78$\pm$0.44  \\
% \textit{Avg. 5 (pipes)} & 6.53$\pm$3.62 & 4.37$\pm$3.07 & 3.81$\pm$2.56 & 3.60$\pm$1.86  \\
% %\hline
% \end{tabular}
% \end{table}

\begin{table}[t] % Example float position options
    \caption{Candidate-to-leak distance error for (I) Modena and (II) L-TOWN-A, comparing GHR-S, GSI, UKF-AW-GSI and FGLL.}
    \centering
    \resizebox{\linewidth}{!}{
    \begin{tabular}{cc|cccc}
        % \hline
        & \multicolumn{1}{c|}{\multirow{2}{*}{\textbf{Metric}}} & \multicolumn{4}{c}{\textbf{Method}} \\
         % \cline{3-6}
        & \multicolumn{1}{c|}{}& \textbf{GHR-S}\rule{0pt}{2.5ex} & \textbf{GSI} & \textbf{UKF-AW-GSI} & \textbf{FGLL}  \\
         \hline
        \multirow{4}{*}{(I)}
        & \textit{Best (km)} &
        \rule{0pt}{2ex} 1.47$\pm$0.77 & 0.93$\pm$0.80 & 0.69$\pm$0.69 & 0.72$\pm$0.51 \\
        % \hline
        & \textit{Best (pipes)} & 7.37$\pm$3.92 & 4.38$\pm$3.33 & 3.28$\pm$2.91 & 3.34$\pm$1.99 \\
        & \textit{Avg. 5 (km)} & 1.44$\pm$0.79 & 0.93$\pm$0.68 & 0.80$\pm$0.57 & 0.78$\pm$0.44  \\
        & \textit{Avg. 5 (pipes)} & 7.12$\pm$3.82 & 4.37$\pm$3.07 & 3.81$\pm$2.56 & 3.60$\pm$1.86  \\
        \hline
        \multirow{4}{*}{(II)}
        & \textit{Best (km)} &
        $0.89\pm0.54$ & $ 0.36\pm0.19 $ & $0.32\pm0.17$ & $ 0.34\pm0.15 $ \\
        & \textit{Best (pipes)} &
        $18.88\pm11.81$ & $ 7.5\pm3.82 $ & $6.88\pm3.29$ & $ 7.38\pm2.94 $ \\
        & \textit{Avg. 5 (km)} &
        $0.83\pm0.47$ & $ 0.36\pm0.18 $ & $0.33\pm0.17$ & $ 0.31\pm0.13 $ \\
        & \textit{Avg. 5 (pipes)} &
        $17.61\pm10.14$ & $ 7.52\pm3.42 $ & $6.95\pm3.25$ & $ 6.68\pm2.68 $ \\
        % \bottomrule
        % \hline
        % \hline
    \end{tabular}%
    }
    % \vspace{3pt}
    \label{table:localization_results_L-TOWN_Modena}
\end{table}

\subsection{L-TOWN - Zone A}

As explained in Section \ref{sec:studies}, the evaluation in the L-TOWN benchmark is executed through the 2018 dataset \cite{Vrachimis2022} and only considering Area A. However, the original 2018 dataset, provided by the BattLeDIM2020 organizers, does not include demand data from Area A, because no AMRs are installed within that area in the original challenge. Considering that both UKF-AW-GSI and FGLL exploit demand measurements, the 2018 dataset is re-computed by means of the Dataset Generator\footnote{\url{https://github.com/KIOS-Research/BattLeDIM}}, provided by the organizers. Thus we are able to simulate the 2018 leaks in the same conditions of the original dataset while including demand measurements in Area A. 

% Additional aspects must also be considered.
We consider the existence of an AMR per pressure sensor, locating them in the same nodes, thus 31 AMRs are assumed to be installed. Moreover, a flow sensor is considered to be installed in the pipe connecting Area A and Area B. This is required to compute the actual inflows to the network that are related to Area A, subtracting the inflows to Areas B and C to the inflow of the water inlets.
Water inlets in Area A are hydraulically separated from the rest of the area through a PRV, which maintains a constant head of 75m at the node immediately downstream. In contrast, the reservoirs have a water height of 100m. This drastic mismatch in hydraulic heads between reservoirs and the rest of the network degrades the UKF performance. To mitigate this, the network graph is restructured for the UKF computations: the reservoirs are removed and the nodes immediately downstream of PRVs are redefined as the new “water inlets” each with a head of 76m.

Regarding the specific method parameters for the L-TOWN benchmark, the UKF process and measurement covariance matrices are $Q=I_n$ and $R=10^{-4}I_{2n_s}$ respectively, and $k_{loop}=75$. Additionally, for GHR-S, $f_p = 6$ and $f=22$.

\subsubsection{Estimation}

In order to generate estimation results, the original 2018 dataset is not enough because it only contains sensor data, as we would need information of the heads in all the network nodes in order to compute estimation error. The size of the benchmark makes it impractical to generate a dataset for all the year that contains all the required information. Therefore, we have decided to generate a specific dataset for each leak (maintaining its characteristics with respect to time of the year, size, type), containing leak information for one hour, i.e., 12 time instants. 

The estimation performances for L-TOWN-A are shown in Table \ref{table:estimation_results_L-TOWN_Modena}. The RMSE between estimated and simulated heads is computed for each one of the compared methods using the aforementioned custom datasets, averaging among the 12 available time instants and the leak scenarios. %Performance results regarding both accuracy and computation time are shown.

The results show how the UKF-based method leads to the best performance in terms of estimation accuracy, but at the cost of a high computation time. In this case, GHR-S produces an extremely fast solution, with a level of accuracy comparable to that of GSI, which requires an extra time in comparison to GHR-S. The balanced performance of GHR-S is caused by the similarity among the nodal heads due to the changes introduced in the benchmark's structure. Regarding FGLL, the estimation error is slightly higher than in the case of the comparison methods, while the computation time is high but much lower than in the case of the UKF-based method.
The FGLL results are again due to the FGO process that produces all the estimates with a single optimization, and therefore the individual estimation error can be higher than in the case of methods which compute estimations individually.

\subsubsection{Localization}

In this case, the pressure and demand data from the 2018 dataset is used to evaluate the methodology. The localization results are presented in Table \ref{table:localization_results_L-TOWN_Modena}. Again, the performance results are represented by the candidate-to-leak distance error in both kilometres and number of pipes. 

The results show that UKF-AW-GSI-LCSM and FGLL outperform both GSI-LCSM and GHR-S-LCSM, with GSI-LCSM producing a comparable results. GHR-S, despite its balanced estimation results, fails in obtaining good localization performance when coupled with LCSM. 
As in the Modena benchmark, we can conclude that UKF-AW-GSI-LCSM is slightly more accurate in terms of node-level localization, while FGLL is more accurate in terms of area-level localization. Again, analysing the estimation performance together with the localization results, we can conclude that FGLL is well suited to obtain precise leak locations despite the higher average estimation error, achieving this around 8 times faster in this benchmark.

%------------------------------------------------------------------
\section{Conclusions and Future Work}\label{sec:conclusion}

%-----------------------------------------------
% TODO -- Future Work
% \textcolor{red}{\begin{itemize}
%     \item Voting methods
%     \item Initializing FGO with UKF data
%     \item Running multiple FGO iterations
%     \item Computation time analysis by varying parameters in FGO - time vs performance (as we did with the UKF) - size of the time window (for instance, T=100 as a whole, or T=50 and sliding)
%     \item show update of the estimation while the FGO window shifts (old estimation is recomputed when new information arrives)
%     \item Demand as another localization index: start solving the standard FGO, then get the l1 indices, impose the leak on the max likelihood node in the demand localization constraint, run again...
%     \item Review the problem associated to the pumping or active elements.
%     \item Multi-leak L-TOWN.
% \end{itemize}}
%-----------------------------------------------

Our paper introduced an FGO-based leak localization architecture, referred to as FGLL, which is the first to our knowledge to analyze the application of FGO in water distribution networks.
FGLL starts from an initial or prior state
coupled with the head and demand sensor readings received
within a time-window.
Each sensor reading generates a new state.
Based on these data points sensor-fusion is performed in order to perform
leak-free state estimation and leak localization through two separate factor graphs.
The optimization process estimates all the states in the current time-window,
unlike most algorithms which focus soley on the current state.

We compare our work with the GHR-S, GSI and UKF-AW-GSI algorithms
on the T-example, Modena, L-TOWN networks, leading to localization leak error metrics in terms of
kilometres and pipe distance.
Our simulations show an improved estimation localization performance of FGLL in comparison with GHR-S-LCSM (45-50\% and 62-63\% of average error reduction, regarding node and area level localization, and the benchmarks of Modena and L-TOWN, respectively) and GSI-LCSM (16-23\% and 6-14\% of average error reduction). In comparison to UKF-AW-GSI-LCSM, the performance is quite equilibrated, with the UKF-based methods outperforming FGLL in terms of node-level localization (4.35\% and 6.25\% error reduction advantage), but with FGLL outperforming UKF-AW-GSI-LCSM in terms of area-level localization (2.5\% and 6.1\% error reduction improvement), with computation times 8 and 15 times lower for FGLL in comparison to UKF-AW-GSI, considering Modena and L-TOWN-A respectively.

In the future, we plan to deepen our study to include
the application of FGO for other water network problems, such as the optimization of the operation of multi-resource energy systems, combining electrical grids, water networks, and heating/cooling systems \cite{Ramos2018}.
For the problem at hand we will expand our investigation
to include running more than a single FGO iteration
in order to obtain improvements in estimation and localization errors
against optimizing the time-window size.
Also of interest is to include the demand as an extra localization factor and tie it to the zero-sum constraint.
Finally, we will perform a deep analysis of the localization performance when real-world phenomena of WDNs, such as multi-leak scenarios and the operation of active elements such as pumps or valves, are considered within the FGO methods.
% \end{itemize}

%------------------------------------------------------------------
\bibliographystyle{IEEEtran}
\bibliography{references}
%------------------------------------------------------------------
\begin{IEEEbiography}[{\includegraphics[width=1in,height=1.25in,clip,keepaspectratio]{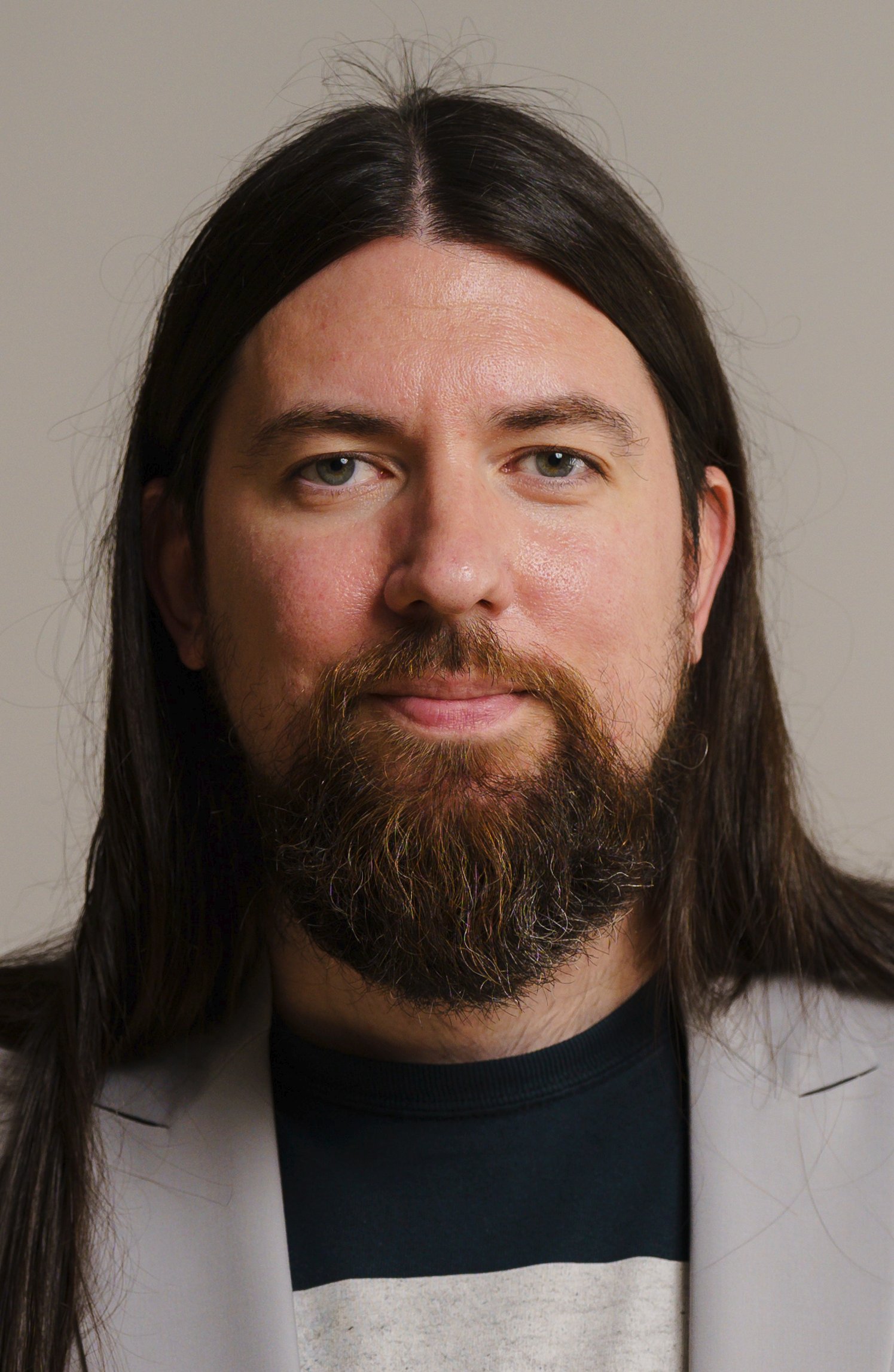}}]{Paul Irofti}
is a Full Professor
within the Computer Science Department
of the Faculty of Mathematics and Computer Science
at the University of Bucharest
He is the co-author of the book “Dictionary Learning Algorithms and Applications” (Springer 2018)
awarded by the Romanian Academy.
He is PhD in Systems Engineering at the Politehnica University of
Bucharest since 2016.
His interests are anomaly detection, signal processing, 
numerical algorithms and optimization.
\end{IEEEbiography}

\begin{IEEEbiography}[{\includegraphics[width=1in,height=1.25in,clip,keepaspectratio]{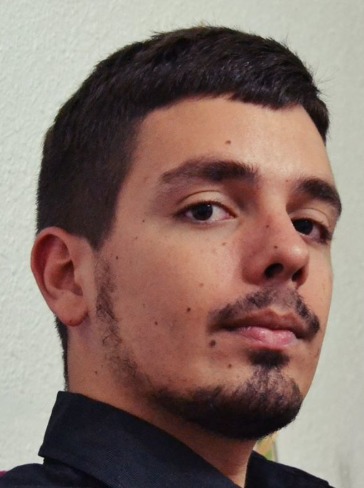}}]{Luis Romero-Ben}
is researcher at the Automatic Control group at the Institut de Robòtica e Informàtica Industrial (CSIC-UPC), Spain. He is a PhD in Real-Time Monitoring of Water Systems by the Universitat Politècnica de Barcelona. His interests are focused on the development and application of data-driven methodologies for the control and monitoring of water distribution networks and urban drainage systems.
\end{IEEEbiography}

\begin{IEEEbiography}[{\includegraphics[width=1in,height=1.25in,clip,keepaspectratio]{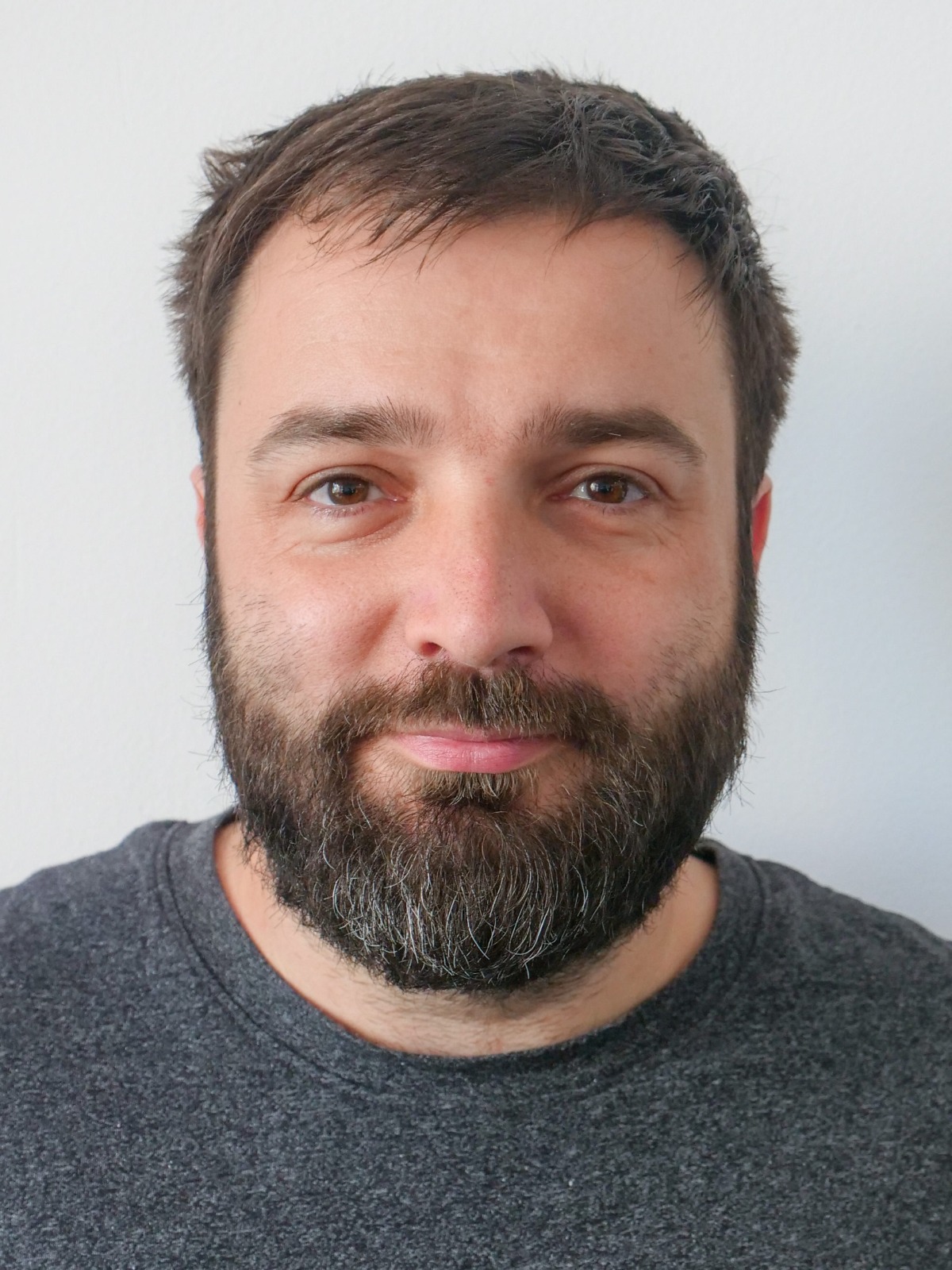}}]{Florin Stoican}
 is Professor in the department of Automatic Control and Systems Engineering, Politehnica University of Bucharest. He obtained his PhD in Control Engineering in 2011 from Supelec (now CentraleSupelec), France with an application of set-theoretic methods for fault detection and isolation. His interests are constrained optimization control, set theoretic methods, fault tolerant control, mixed integer programming, motion planning.
\end{IEEEbiography}

\begin{IEEEbiography}[{\includegraphics[width=1in,height=1.25in,clip,keepaspectratio]{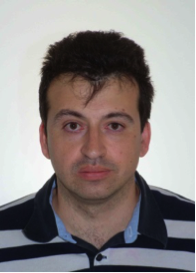}}]{Vicenç Puig}
 holds a PhD in Automatic Control, Vision and Robotics and is the leader of the research group Advanced Control Systems (SAC) at the Polytechnic University of Catalonia.
  He has important scientific contributions in the areas of fault diagnosis and fault tolerant control using interval models. He participated in more than 20 international and national research projects in the last decade. He led many private contracts, and published more than 80 articles in JCR journals and more than 350 in international conference/workshop proceedings. Prof. Puig supervised over 20 PhD theses and over 50 MA/BA theses.
\end{IEEEbiography}

\vfill

\end{document}